\documentclass[aps,pra,twocolumn]{revtex4-2}

\usepackage[usenames,dvipsnames]{color}
\usepackage[utf8]{inputenc}
\usepackage{lipsum}
\usepackage{graphicx}
\usepackage{amssymb}
\usepackage{subfigure}
\usepackage{stackrel}
\usepackage{enumitem}
\usepackage{amsmath}
\usepackage{amssymb}
\usepackage{amsthm}
\usepackage{amsfonts}
\usepackage{bbm}
\usepackage{bm}
\usepackage{color}
\usepackage{verbatim} %for comments
\usepackage{tcolorbox}
\usepackage[normalem]{ulem}
\usepackage{hyperref}
\usepackage{soul}

\newlist{subquestion}{enumerate}{1}
\setlist[subquestion,1]{label=(\alph*)}

\newcommand{\vect}[1]{\ensuremath{\bm{{#1}}}}
\newcommand{\ket}[1]{\ensuremath{\left|{#1}\right\rangle}}
\newcommand{\bra}[1]{\ensuremath{\left\langle{#1}\right |}}

\newcommand{\diego}[1]{ #1}

\newcommand{\Trr}[1]{\textrm{Tr}\left[#1\right]}
\newcommand{\TrP}[2]{\textrm{Tr}_{#1}\left[#2\right]}

\newcommand{\beq}{\begin{equation}}
\newcommand{\eeq}{\end{equation}}
\newcommand{\bse}{\begin{subequations}}
\newcommand{\ese}{\end{subequations}}\newcommand{\bea}{\begin{eqnarray}}
\newcommand{\eea}{\end{eqnarray}}
\newcommand{\bit}{\begin{itemize}}
	\newcommand{\eit}{\end{itemize}}
\newcommand{\bpmatrix}{\begin{pmatrix}}
	\newcommand{\epmatrix}{\end{pmatrix}}

\newcommand{\be}{\begin{equation}}
\newcommand{\ee}{\end{equation}}
\newcommand{\ben}{\begin{eqnarray}}
\newcommand{\een}{\end{eqnarray}}

\newcommand{\QJSD}{\text{QJSD}}

\newcommand{\Sr}[2]{\text{S}_r\!\left(#1||#2\right)}
\newcommand{\av}[1]{\overline{#1}}

\newtheorem{proposition}{Proposition}

\begin{document}

\title{Transmission distance in the space of quantum channels}

\author{Diego G. Bussandri$^{1,2}$, Pedro W. Lamberti$^{2,3,4}$, Karol $\dot{\text{Z}}$yczkowski$^{4,5}$}
\affiliation{$^1$Instituto  de  F\'isica  La  Plata  (IFLP)  and  Departamento  de  F\'isica,  Facultad  de Ciencias Exactas, Universidad Nacional de La Plata, C.C. 67, 1900 La Plata, Argentina}
\affiliation{$^2$Consejo Nacional de Investigaciones Científicas y Técnicas de la República Argentina (CONICET), Av. Rivadavia 1917, C1033AAJ, CABA, Argentina}
\affiliation{$^3$Facultad de Matem\'atica, Astronom\'{\i}a, F\'{\i}sica y Computaci\'on, Universidad Nacional de C\'ordoba, \\ Av. Medina Allende s/n, Ciudad Universitaria, X5000HUA C\'ordoba, Argentina}
\affiliation{$^4$Faculty of Physics, Astronomy and Applied Computer Science, Institute of Theoretical Physics, Jagiellonian University, ul. {\L}ojasiewicza 11, 30–348, Kraków, Poland}
\affiliation{$^5$Center for Theoretical Physics, Polish Academy of Sciences, Al. Lotników 32/46, 02-668 Warszawa, Poland}

\begin{abstract}
We analyze two ways to obtain distinguishability measures between quantum maps by employing the square root of the quantum Jensen-Shannon divergence, which forms a true distance in the space of density operators. The arising measures are the transmission distance between quantum channels and the entropic channel divergence. We investigate their mathematical properties and discuss their physical meaning. Additionally, we establish a
chain rule for the entropic channel divergence, which implies the amortization collapse, a relevant result with potential applications in the field of discrimination of quantum channels and converse bounds. 
Finally, we analyze the distinguishability between two given Pauli channels and study exemplary Hamiltonian
dynamics under decoherence.
\end{abstract} 

\date{April 6, 2023}
\pacs{}
\maketitle

\section{Introduction}\label{sec:Intro}

\noindent The notion of quantum channel distinguishability is at the core of quantum information theory, and it plays a central role in a variety of contexts. Different works investigate the mathematical and physical conditions for a suitable measure of distance between quantum maps and, correspondingly, various such measures have been introduced, with trace distance and quantum fidelity being the most widely used \cite{Gilchrist2005}. %\diego{
Constructing a universal distance measure
in the space of quantum maps
that fulfils all the suitable requirements is strongly motivated by the recent literature. 
However,  finding  such a 
\textit{gold standard} 
is rather difficult \cite{Puchaa2011},
and one tries to identify 
distance measures capable to
compare 
theoretically idealized quantum channels
 with their noisy experimental implementations.

Within the list of relevant requirements for a
measure studied, an important property is the triangle inequality, as it allows one
to construct a true distance
and it  serves as a tool to establish other
features, including the chaining property. Recently, Virosztek \cite{Virosztek2021} and Sra \cite{Sra2021} demonstrated that the square root of the quantum Jensen-Shannon divergence (QJSD),  
satisfies the triangle inequality for 
any quantum states of an arbitrary
finite dimension.
%and in the cone of positive matrices. 
This extensively used entropic distinguishability measure
has appealing properties %and interpretations, 
and it has been widely used in quantum information theory \cite{Audenaert2014,Radhakrishnan2016,Megier2021,Settimo2022}.

The main aim of this work is to extend 
the transmission distance,
defined as square root of 
the quantum Jensen-Shannon divergence \cite{Briet2009}, to the space of quantum channels.
We study two different approaches to carry out this goal:
Making use of the Choi--Jamio\l{}kowski isomorphism,
we arrive at the
\textit{transmission distance between quantum channels}. Furthermore,
by optimizing the channel output 
over all possible inputs, 
we investigate the \textit{entropic channel divergence}. 

Going beyond the required properties for having well-behaved measures of distance between quantum operations, we establish a chain rule for the entropic channel divergence. This chain rule was originally proposed in Eq. (4) of Ref. \cite{Fang2020} for the quantum relative entropy, motivated by its classical counterpart. \diego{However, the extension of the quantum relative entropy 
	to the space of quantum maps through optimization of its inputs does not satisfy this particular chain rule.}
%and a similar one was demonstrated for the regularized version of the channel relative entropy. 

We address the issue of the \textit{amortized distinguishability of quantum channels}, relevant to analyze the problem of hypothesis testing for quantum channels \cite{Wilde2020a}. The idea behind amortized distance measures is to consider two quantum states as inputs of two different quantum channels to explore the biggest distance between these channels without considering the original distinguishability that the input states may have. The chain rule 
%mentioned before immediately 
leads to another property called \textit{amortization collapse} \cite{Wilde2020a}, which occurs if the channel divergence is equal to its amortized version. \diego{In such a case, one obtains useful single-letter converse bounds on the capacity of adaptive channel discrimination protocols \cite{Leditzky2018a}. }

Finally, we will examine two specific applications for the entropic distinguishability measures: a) Pauli channels, with a focus on studying noise
in the standard quantum teleportation channel \cite{Shahbeigi2018}; and   b) the distinguishability of Hamiltonians under decoherence, a particular case within the discrimination of superoperators proposed in \cite{Childs2000,Raginsky2001}.  

This paper is organized as follows. In Sec. \ref{sec:prelim} we summarize the main properties of the transmission distance in the space of quantum states. In Sec. \ref{sec:transmissiondistancebetQC} we introduce the transmission distance between quantum channels through the Choi--Jamio\l{}kowski isomorphism and study its properties.  The entropic channel divergence is proposed
and analyzed in Sec. \ref{sec:entropicdisting}. 

The chain rule and the amortization collapse of the entropic channel divergence are presented in Sec. \ref{sec:chainrule} and in Sec. \ref{sec:connection} we consider a set of quantum maps, for which the proposed measures are equal. In Sec. \ref{sec:PhysInt} the physical motivations and operational meanings of the introduced distances is discussed. In Sec. \ref{sec:applications}, we compute analytically the distances for Pauli channels and for arbitrary Hamiltonians under decoherence. Sec. \ref{sec:concludingR} concludes the article with a brief review of results obtained.

\section{QJSD and transmission distance in the space of quantum states}\label{sec:prelim}

\noindent Let $\mathcal{M}_N$ be the space of density matrices $\rho$ (positive and normalized operators, $\rho \geq 0$ and ${\rm Tr} \rho =1$, respectively) defined on a %separable
% in our context this mathematical term can be 
% misleading
$N$-dimensional Hilbert space. 

The von Neumann entropy,
$\text{S}(\rho)=-\Trr{\rho \log_2 \rho}$,
satisfies the concavity property \cite{M.Ohya} %\cite{OP93}
\begin{align}
	\text{S}(\av{\rho})\geq \sum_{i}p_i\text{S}(\rho_i),
\end{align}
for a given ensemble of quantum states $\{p_i,\rho_i\}_i$, with the weighted average
$\av{\rho}=\sum_{i}p_i\rho_i$. This property gives rise to a suitable symmetric measure of distinguishability between the states composing the ensemble (according to the classical probability vector $\vect{p}=\{p_i\}_i$) called 
{\sl Holevo quantity} \cite{Holevo1973,Holevo2012} %\cite{Ho72,Ho73,HG12}
or  quantum Jensen-Shannon divergence
\cite{Majtey2005a,Lamberti2008,Briet2009,Virosztek2021,Sra2021},
\begin{align}
	\text{QJSD}_{\vect{p}}(\rho_1,\dots,\rho_n)=\text{S}(\av{\rho})-\sum_{i}p_i\text{S}(\rho_i).
\end{align}
Making use of the quantum relative entropy \cite{M.Ohya}
between two states $\rho$ and $\sigma$, 
\begin{align}\label{eq:quantumrelativeentropy}
	\text{S}_r(\rho||\sigma)=\Trr{\rho(\log_2\rho - \log_2 \sigma)},
\end{align}
the quantum divergence  can be recast in the form 
\begin{align}
	\text{QJSD}_{\vect{p}}(\rho_1,\dots,\rho_n)&=\sum_i p_i \Sr{\rho_i}{\av{\rho}}.
\end{align}
This equality allows us to interpret the quantity $\text{QJSD}_{\vect{p}}(\rho_1,\dots,\rho_n)$ as \textit{total divergence to the average} (or \textit{information radius}) quantifying how much information is discarded if we describe the system employing just the convex combination $\av{\rho}=\sum_i p_i \rho_i$. An analogous interpretation can be given in the classical setup \cite{Manning2011,Manning2002}.

In the case of a 
binary ensemble of states $\rho$ and $\sigma$ combined  with equal weights,
 we can employ a simplified notation, 
\begin{align}
	\QJSD(\rho,\sigma)&=\text{S}\!\left(\frac{\rho+\sigma}{2}\right)-\frac{1}{2}\text{S}(\rho)-\frac{1}{2}\text{S}(\sigma).
\end{align}

Regarding mathematical properties, the QJSD satisfies the
\textit{indiscernibles identity} \cite{Lamberti2008},
%namely 
\begin{align}
	&0\leq \text{QJSD}(\rho,\sigma) \leq 1 \ \text{ with } \nonumber \\
	&\text{QJSD}(\rho,\sigma)=0 \  \iff  \ \rho=\sigma \nonumber \\
	&\textrm{QJSD}(\rho,\sigma)= 1 \ \iff  \ \text{supp}(\rho)\perp\text{supp}(\sigma),\label{eq:indisId}
\end{align}
where $\text{supp}(\rho)\perp\text{supp}(\sigma)$ denotes $\rho$ and $\sigma$ with orthogonal supports.

 The quantum relative entropy  satisfies the monotonicity \cite{M.Ohya}
with respect to  any  completely positive trace preserving  (CPTP) map $\Phi$.
This property, also called  
\textit{data processing inequality} \cite{Lamberti2008},
is thus inherited by the  quantum divergence,
\begin{equation}
\QJSD(\Phi\rho,\Phi\sigma) \; \leq \; \textrm{QJSD}(\rho,\sigma).
\label{eq:dpineq}
\end{equation}

 Furthermore, monotonicity implies that QJSD satisfies the restricted additivity,
\begin{align}\label{eq:restradd}
\textrm{QJSD}(\rho_1\otimes\sigma , \rho_2\otimes\sigma)=\textrm{QJSD}(\rho_1 , \rho_2),
\end{align}
and the invariance with respect to an arbitrary unitary transformation $U$
acting on both states,
$$\QJSD(U\rho U^\dagger,U\sigma U^ \dagger) = \textrm{QJSD}(\rho,\sigma).$$

\diego{In the single qubit case, $N=2$,
	Bri{\"e}t and Harremo{\"e}s showed}
 \cite{Briet2009} 
 that the square root of the $\QJSD$, known as the \textit{transmission distance},
 \begin{align}
 \label{eq:tranmissiondistance}
	d_{\text{t}}(\rho,\sigma)\ :=\ \sqrt{\QJSD(\rho,\sigma)},
\end{align}
satisfies the triangle inequality,
\begin{align}\label{eq:triangleineqstates}
	d_{\text{t}}(\rho,\sigma)\leq d_{\text{t}}(\rho,\chi) + d_{\text{t}}(\chi,\sigma).
\end{align}
for any $\rho,\sigma,\chi \in \mathcal{M}_2$. Recently, this result has been established
 for an arbitrary finite dimension $N$
 and extended to the cone of positive matrices \cite{Virosztek2021,Sra2021}.

The transmission distance can be bounded
by other known distance measures. 
For instance, the \textit{trace distance}  $\textrm{T}(\rho,\sigma)=\frac{1}{2}\Trr{\sqrt{(\rho-\sigma)^2}}$,
allows one to obtain the bounds 
\begin{align}\label{eq:IneqJS}
	\frac{\textrm{T}(\rho,\sigma)}{\sqrt{2\log 2}}\leq d_{\text{t}}(\rho,\sigma)\leq \sqrt{\textrm{T}(\rho,\sigma)},
\end{align}
valid for an arbitrary dimension $N$. The upper bound was derived
 in \cite{Briet2009}, while  the lower one
follows from inequalities \cite{Audenaert2014},
\begin{align*}
	2(1-\alpha)^2\textrm{T}(\rho,\sigma)^2\leq \Trr{\rho (\log \rho - \log \overline{\rho}_\alpha )},
\end{align*}
with $\overline{\rho}_\alpha=\alpha \rho +(1-\alpha)\sigma$ and $0<\alpha<1$. Inserting $\alpha=1/2$, one arrives at
\begin{align*}
\frac{\textrm{T}(\rho,\sigma)^2}{2\log 2}&\leq \text{S}_r(\rho||\frac{\rho+\sigma}{2}) \text{ and } \\
\frac{\textrm{T}(\sigma,\rho)^2}{2\log 2}&\leq \text{S}_r(\sigma||\frac{\rho+\sigma}{2}).
\end{align*}
The constant $\log 2$ appears above
as the quantum relative entropy \eqref{eq:quantumrelativeentropy}
is defined here with logarithm base two.
Therefore, we obtain
\begin{align*}
\frac{\textrm{T}(\rho,\sigma)^2}{2\log 2}&\leq \frac{1}{2}\text{S}_r(\rho||\frac{\rho+\sigma}{2})+\frac{1}{2}\text{S}_r(\sigma||\frac{\rho+\sigma}{2}),
\end{align*}
and by taking square root we
arrive at the lower bound in inequality \eqref{eq:IneqJS}.

A complementary upper bound for the transmission distance in terms of the square root of the quantum fidelity,
$F(\rho,\sigma)=(\Trr{\sqrt{\sqrt{\rho}\sigma\sqrt{\rho}}})^2$,
\begin{align}\label{eq:IneqJS2}
	\sqrt{\textrm{QJSD}(\rho,\sigma)}\leq D_E(\rho,\sigma), 
\end{align}
was established in \cite{Roga2010}.
The quantity $D_E$ is called the {\sl entropic distance} \cite{Lamberti2008},
\begin{align}\label{eq:boundK}
D_E(\rho,\sigma)=\sqrt{H_2\left\{\frac{1}{2}\left[1-\sqrt{F(\rho,\sigma)}\right]\right\}},
\end{align}
as it is a function of the binary entropy,
$H_2(x)=-x\log_2 x-(1-x)\log_2(1-x)$
for $x\in [0,1]$.

%arXiv:2301.09192
% Lower Bounds on Learning Pauli Channels
% O. Fawzi, A. Oufkir, D. Stilck França 
%%%   distinguishing between Pauli channels ....
% J. Watrous,  (2018). The theory of quantum information. Cambridge university press

\section{Transmission distance between quantum channels and Jamio\l{}kowski isomorphism}\label{sec:transmissiondistancebetQC}

\noindent In the preceding section, we recalled the transmission distance in the space $\mathcal{M}_N$ of quantum states.	Let us introduce now a measure of distinguishability between completely positive trace-preserving maps, $\mathcal{E}:\mathcal{M}_N \to \mathcal{M}_N,$ by using the Choi-Jamio\l{}kowski isomorphism which establishes a one-to-one correspondence
between a quantum operation $\mathcal{E}$ and 
the corresponding bipartite quantum state $\rho_{\mathcal{E}}$ \cite{Watrous2018},
\begin{align}\label{eq:Choistate}
	\rho_{\mathcal{E}}=(\mathcal{E}\otimes\mathbbm{1})(\ket{\Phi}\bra{\Phi}).
\end{align}
Here 
\begin{align}\label{eq:Bellstate}
	\ket{\Phi}=\sum_i\frac{1}{\sqrt{N}} \ket{i}_{a}\ket{i}_{b},
\end{align}
denotes the maximally entangled, generalized Bell state,
represented in some orthonormal basis
$\{\ket{i}_x\}_i \in {\mathcal{H}}_N$ 
of the $N$-dimensional Hilbert space.
The bipartite state $\rho_{\mathcal{E}}$ is called the \textit{Choi state} of the map $\mathcal{E}$ and represents a mixed state in $\mathcal{M}_{N^2}$.
It emerges  by applying $\mathcal{E}$ to the 
principal system, maximally entangled with an 
ancilla of the same dimension $N$. 

Making use of this isomorphism, we apply
Eq. \eqref{eq:tranmissiondistance}
to define the \textit{transmission distance between channels} $\mathcal{E}$ and $\mathcal{F}$, 
	\begin{align}
	\label{eq:ChoiDisting}
	d_{\textrm{t}}^{\text{iso}}(\mathcal{E},\mathcal{F})\; := \; d_{\textrm{t}}(\rho_{\mathcal{E}},\rho_{\mathcal{F}}).
	\end{align}
	
Instead of QJSD we use its square root $d_{\textrm{t}}$
to assure that the triangle inequality is satisfied \cite{Virosztek2021}
and Eq. 	(\ref{eq:ChoiDisting})
can serve as a metric between quantum maps \cite{Gilchrist2005}.

\subsection{Properties of $d_{\textrm{t}}^{\text{iso}}(\mathcal{E},\mathcal{F})$}\label{sec:propChoiDist}
\noindent 
A list of required properties for a suitable measure of distinguishability between quantum maps was discussed in  \cite{Raginsky2001,Gilchrist2005,Puchaa2011}.
Let us now verify, which of them are satisfied by the distance
$d_{\textrm{t}}^{\text{iso}}(\mathcal{E},\mathcal{F})$.

Since the triangle inequality \eqref{eq:triangleineqstates}
is satisfied for the transmission distance in the state space,
the quantity $d_{\textrm{t}}^{\text{iso}}(\mathcal{E},\mathcal{F})$
is symmetric in its arguments, it satisfies the triangular inequality,
is non-negative and vanishes if and only if $\mathcal{E}=\mathcal{F}$). 
Hence $d_{\textrm{t}}^{\text{iso}}(\mathcal{E},\mathcal{F})$
forms a true distance in the space of quantum maps.

For this kind of measures one often requires 
their \textit{stability} with respect to the tensor product,
\begin{align}
d_\text{t}^{\text{iso}}(\mathcal{E}\otimes\mathbbm{1},\mathcal{F}\otimes\mathbbm{1})=d_\text{t}^{\text{iso}}(\mathcal{E},\mathcal{F}).
\end{align} 
This fact can be demonstrated employing the restricted additivity \eqref{eq:restradd}, and relation $\rho_{\mathcal{E}\otimes\mathbbm{1}}=\rho_{\mathcal{E}}\otimes\rho_\mathbbm{1}$, which yield
\begin{align*}
	d_{\textrm{t}}^{\text{iso}}(\mathcal{E}\otimes\mathbbm{1},\mathcal{F}\otimes\mathbbm{1})&=\sqrt{\textrm{QJSD}(\rho_\mathcal{E}\otimes\rho_\mathbbm{1},\rho_\mathcal{F}\otimes\rho_\mathbbm{1})}\\
	&=\sqrt{\textrm{QJSD}(\rho_\mathcal{E},\rho_\mathcal{F})}=d_{\textrm{t}}^{\text{iso}}(\mathcal{E},\mathcal{F}).
\end{align*}

Another property of \textit{chaining} is relevant to estimate errors in
protocols of quantum information processing. 
It is satisfied by a distance $d$ if for any four maps 
$\mathcal{E}_1,\mathcal{F}_1, \mathcal{E}_2,\mathcal{F}_2$
 the distance between their concatenations
can be bounded from above,
%\begin{small}
\begin{align}
\label{eq:chainingrule1}
	d(\mathcal{E}_2\circ  \mathcal{E}_1,\mathcal{F}_2\circ\mathcal{F}_1)
	\leq 
	 d(\mathcal{E}_1,\mathcal{F}_1)+d(\mathcal{E}_2,\mathcal{F}_2).
\end{align}	
%\end{small}

%{\color{red} Editorial: the discussion for 
%an arbitrary distance $d^{\rm iso}$ related to the isomorphism
%was slightly misleading and redundant,
%so I suggest to remove it. K.} 
%\diego{This is OK for me. I have modified a bit the concluding remarks (in my color) according to this change (Diego)}

In general, this property is not satisfied by
the distance 
$d_{\textrm{t}}^{\text{iso}}(\mathcal{E},\mathcal{F})$
defined \eqref{eq:ChoiDisting}
by the Jamio{\l}kowski isomorphism. 
%
% Jamio\l{}kowski isomorphism. Let us consider
%\begin{align}\label{eq:ChoiDistingGeneral}
%	d^{\text{iso}}(\mathcal{E},\mathcal{F})=d(\rho_{\mathcal{E}},\rho_{\mathcal{F}}),
%\end{align}
%being $d(\cdot,\cdot)$ an arbitrary distance measure which satisfies the \textit{indiscernibles identity}, Eq. \eqref{eq:indisId}. 
%Chaining property states that the distance between $\mathcal{E}_2\circ  \mathcal{E}_1$ and $\mathcal{F}_2\circ\mathcal{F}_1$ is lower or equal than the distance between $\mathcal{E}_1$ and $\mathcal{F}_1$ plus the one between $\mathcal{E}_2$ and $\mathcal{F}_2$. However, we can prove that this inequality does not hold for $d^{\text{iso}}$ when we consider arbitrary quantum operations. 
To show a counterexample consider the following 
collection of four selected Choi states analyzed in \cite{Puchaa2011}, 
\begin{align}
	\rho_{\mathcal{E}_1}&=\frac{1}{2}\textrm{diag}(1,1,0,0) \\
	\rho_{\mathcal{E}_2}&=\frac{1}{2}\textrm{diag}(1,0,0,1) \\
	\rho_{\mathcal{F}_1}&=\rho_{\mathcal{E}_1} \\
	\rho_{\mathcal{F}_2}&= \frac{1}{2}\textrm{diag}(0,0,1,1). 
\end{align}

%{\color{red} Editorial: are all the primes in this paragraph really necessary?
%Can we simplify notation by changing say 
%${\mathcal{E}'_i} \to {\mathcal{E}_i}$. \\
 %K.} 
 %\diego{Diego: Yes, we could take ${\mathcal{E}'_i} \to {\mathcal{E}_i}$ as well. I have included the primes $\mathcal{E}'_i$ just to distinguish these particular channels from the arbitrary ones denoted with $\mathcal{E}_i$, but if from your point of view taking ${\mathcal{E}'_i} \to {\mathcal{E}_i}$ is clearer, then I would do it.} 
% THANKS for the remark!
%% simplification introduced - primes are removed !
Hence $\rho_{\mathcal{E}_2\circ\mathcal{E}_1}=\rho_{\mathcal{E}_1}$ and $\rho_{\mathcal{F}_2\circ\mathcal{F}_1}=\rho_{\mathcal{F}_2}$, 
so the transmission distance between both composed maps reads,
$$d^{\text{iso}}_{t}(\mathcal{E}_2\circ  \mathcal{E}_1,\mathcal{F}_2\circ\mathcal{F}_1)=d^{\text{iso}}_t(\mathcal{E}_1,\mathcal{F}_2).$$
As the Choi states $\rho_{\mathcal{E}_1}$ and $\rho_{\mathcal{F}_2}$ have orthogonal supports,
the distance
%$d^{\text{iso}}_t(\mathcal{E}'_2\circ  \mathcal{E}'_1,\mathcal{F}'_2\circ\mathcal{F}'_1)$ 
$d^{\text{iso}}_t(\mathcal{E}_1,\mathcal{F}_2)=1$,
as it 
admits the maximal value of implied the identity of indiscernibles \eqref{eq:indisId}.
Since 
	$\rho_{\mathcal{F}_1}=\rho_{\mathcal{E}_1}$
one has
 $$d^{\text{iso}}_{t}(\mathcal{E}_1,\mathcal{F}_1)+
d^{\text{iso}}_{t}(\mathcal{E}_2,\mathcal{F}_2)=
d^{\text{iso}}_{t}(\mathcal{E}_2,\mathcal{F}_2).$$ 
Taking into account that $\rho_{\mathcal{E}_2}$ and $\rho_{\mathcal{F}_2}$ do not have orthogonal supports, \diego{we obtain the inequality,} 
\begin{align*}
    d^{\text{iso}}_{t}(\mathcal{E}_2\circ  \mathcal{E}_1,\mathcal{F}_2\circ\mathcal{F}_1)&>
d^{\text{iso}}_{t}
(\mathcal{E}_2,\mathcal{F}_2)=\\
&=d^{\text{iso}}_{t}(\mathcal{E}_1,\mathcal{F}_1)+d^{\text{iso}}_{t}(\mathcal{E}_2,\mathcal{F}_2),
\end{align*}
which provides a counterexample of inequality 
\eqref{eq:chainingrule1}.

%Therefore, the chaining property is not fulfilled for an arbitrary distance $d(\cdot,\cdot)$ which satisfies Eq. \eqref{eq:indisId}.
%\diego{As a result, when the Choi matrices of the maps are used to define distinguishability measures,  as in Eq. \eqref{eq:ChoiDisting}, the above conclusion suggests that the resulting distance measure may constitute a candidate for error or diagnostic measure (i.e. to evaluate the performance of experimental -noisy- channels from its ideal implementation), insted of a gold-standard one.}

However, the chaining property holds in a particular case,
if one of the maps applied first,  $\mathcal{E}_1$ or  $\mathcal{F}_1$,
is bistochastic: 
trace-preserving and unital.
 As a consequence of the monotonicity of the transmission distance and the triangle inequality, the chaining property holds  
 for a bistochastic argument,
  $\mathcal{F}_1=\mathcal{D}_{\text{bi}}$.
  To demonstrate  the desired inequality,
%\begin{small}
\begin{align}\label{eq:chainingrule}
	d_{\text{t}}^{\text{iso}}(\mathcal{E}_2\!\circ \! \mathcal{E}_1,\mathcal{F}_2 \! \circ \! \mathcal{D}_{\text{bi}})\leq d_{\text{t}}^{\text{iso}}(\mathcal{E}_1,\mathcal{D}_{\text{bi}})+d_{\text{t}}^{\text{iso}}(\mathcal{E}_2,\mathcal{F}_2),
\end{align}	
%\end{small}
we follow directly the same steps as in Ref. \cite{Gilchrist2005}. 
%{\color{red} Diego - please add a short explanation of this fact here ! Diego: Sure}
%\diego{
By applying the triangle inequality, we have
\begin{align}\nonumber
    d_{\text{t}}^{\text{iso}}(\mathcal{E}_2\!\circ\!  \mathcal{E}_1,\mathcal{F}_2\!\circ\!\mathcal{D}_{\text{bi}})&
\leq d_{\text{t}}^{\text{iso}}(\mathcal{E}_2\!\circ\! \mathcal{E}_1 , \mathcal{E}_2\!\circ\! \mathcal{D}_{\text{bi}})\\&+d_{\text{t}}^{\text{iso}}(\mathcal{E}_2 \!\circ\! \mathcal{D}_\text{bi}, \mathcal{F}_2 \!\circ\! \mathcal{D}_\text{bi}).\label{eq:ChainingProve}
\end{align}
Note that for arbitrary operations $\mathcal{E}$ and $\mathcal{F}$ it holds $\rho_{\mathcal{E}\!\circ\!\mathcal{F}}=(\mathcal{F}^{\intercal}\otimes\mathcal{E})(\ket{\Phi}\bra{\Phi})$,
where $\mathcal{F}^{\intercal}$ denotes the adjoint quantum operation:
if  $\{F_i\}_i$ represents Kraus operators corresponding to the map $\mathcal{F}$, their adjoints, $\{F_i^\intercal\}_i$ 
determine $\mathcal{F}^{\intercal}$. 
 If $\mathcal{F}$ is a unital map,
 its adjoint $\mathcal{F}^\intercal$ is trace-preserving, 
 and thus,
\diego{\begin{align*}
	&d_{\text{t}}^{\text{iso}}(\mathcal{E}_2 \!\circ\! \mathcal{D}_\text{bi}, \mathcal{F}_2 \!\circ\! \mathcal{D}_\text{bi})=\\&d_{\text{t}}\Big[(\mathcal{D}_\text{bi}^{\intercal} \otimes \mathcal{E}_2)(\ket{\Phi}\bra{\Phi}),(\mathcal{D}_\text{bi}^{\intercal} \otimes \mathcal{E}_2)(\ket{\Phi}\bra{\Phi})\Big].
	\end{align*}}

\noindent Therefore, the right-hand side of Eq. \eqref{eq:ChainingProve} can be bounded by employing contractivity to both terms, leading to the desired result. 
%}

The post-processing inequality \cite{Gilchrist2005,Raginsky2001} requires that 
	\begin{equation}
	\label{eq:postproc}
d_{\textrm{t}}^{\text{iso}}(\mathcal{R}\!\circ\! \mathcal{E},\mathcal{R}\!\circ\!\mathcal{F})
\; \leq \;
d_{\textrm{t}}^{\text{iso}}(\mathcal{E},\mathcal{F}),
\end{equation}
for arbitrary quantum maps $\mathcal{R}$, $\mathcal{E}$ and $\mathcal{F}$. 
The transmission distance $d^{\text{iso}}_{t}$
satisfies this property, as it follows
from the monotonicity of this distance.  
%%%%%
	
Inequality \eqref{eq:chainingrule} and post-processing inequality
\eqref{eq:postproc} allow us 
to demonstrate the  invariance with respect
to arbitrary unitary operations
$\mathcal{U}$ and $\mathcal{V}$,
	\begin{align}\label{eq:unitaryinv}
	d_{\textrm{t}}^{\text{iso}}(\mathcal{U}\!\circ\!  \mathcal{E} \!\circ\! \mathcal{V},\; \mathcal{U}\!\circ\!  \mathcal{F}\!\circ\! \mathcal{V})= d_{\textrm{t}}^{\text{iso}}(\mathcal{E},\mathcal{F}). 
\end{align}
 Note that $d_{\textrm{t}}$ is invariant under a post-transformation of $\mathcal{E}$ with $\mathcal{U}$, 
$$d_{\textrm{t}}^{\text{iso}}(\mathcal{U}\!\circ\!  \mathcal{E} ,\mathcal{U}\!\circ\!  \mathcal{F})= d_{\textrm{t}}^{\text{iso}}(\mathcal{E},\mathcal{F}),$$
because of the unitary invariance of the transmission distance in the state space. Thus, it remains to show the identity,
\begin{align}
\label{eq:unitary2}
d_{\textrm{t}}^{\text{iso}}(\mathcal{E} \!\circ\! \mathcal{V},\mathcal{F}\!\circ\! \mathcal{V})= d_{\textrm{t}}^{\text{iso}}(\mathcal{E},\mathcal{F}).
\end{align}
The chaining property in this case states that
$$d_{\textrm{t}}^{\text{iso}}(\mathcal{E} \!\circ\! \mathcal{V},\mathcal{F}\!\circ\! \mathcal{V}) \leq d_{\textrm{t}}^{\text{iso}}(\mathcal{E},\mathcal{F}).$$
Simultaneously it holds,
$$d_{\textrm{t}}^{\text{iso}}(\mathcal{E},\mathcal{F}) = d_{\textrm{t}}^{\text{iso}}(\mathcal{E}_{\mathcal{V}} \!\circ\! \mathcal{V}^{-1},\mathcal{F}_{\mathcal{V}}\!\circ\! \mathcal{V}^{-1}) \leq d_{\textrm{t}}^{\text{iso}}(\mathcal{E}_{\mathcal{V}},\mathcal{F}_{\mathcal{V}}),$$
where $\mathcal{E}_{\mathcal{V}}=\mathcal{E}\!\circ\! \mathcal{V}$ and $\mathcal{F}_{\mathcal{V}}=\mathcal{F}\!\circ\! \mathcal{V}$. Therefore, we conclude that
$$d_{\textrm{t}}^{\text{iso}}(\mathcal{E},\mathcal{F}) \leq d_{\textrm{t}}^{\text{iso}}(\mathcal{E}_{\mathcal{V}},\mathcal{F}_{\mathcal{V}})\leq d_{\textrm{t}}^{\text{iso}}(\mathcal{E},\mathcal{F}).$$
This implies Eq. \eqref{eq:unitary2}
and completes the proof of the unitary invariance \eqref{eq:unitaryinv}.

To establish bounds on the
analyzed transmission distance
$d^{\text{iso}}_{t}$
we shall apply the Jamio{\l}kowski
isomorphism to extend the standard distance
measures defined in the space of 
states into the space of maps \cite{Roga2011}.
%{RZF11} W. Roga,  K. {\.Z}yczkowski, M. Fannes,
%  Entropic characterization of quantum operations 
% {\sl IJQI} {\bf 9}, 1031-1045 (2011).
The trace distance 	$T$,
fidelity $F$,  Bures distance $D_B$
and the entropic distance $D_E$
between any two maps read, respectively,
\begin{align}
	T(\mathcal{E},\mathcal{F})&=\textrm{T}(\rho_\mathcal{E},\rho_\mathcal{F}),
		\label{eq:tracedistanceChoi}\\
		F(\mathcal{E},\mathcal{F})&=F(\rho_\mathcal{E},\rho_\mathcal{F}), \label{eq:fidelitychoi}\\
	D_B(\mathcal{E},\mathcal{F})&=\sqrt{2-2\sqrt{F(\mathcal{E},\mathcal{F})}},\label{eq:buresdistanceChoi}\\
	D_E(\mathcal{E},\mathcal{F})&=\sqrt{H_2\left\{D_B^2(\mathcal{E},\mathcal{F})/4\right\}} \label{eq:KarolmeasureChoi}.
\end{align}

%{\color{blue} Editorial: 
%standard definition of Bures distance with prefactor 
%$\sqrt{2}$ is used}

Making use of inequalities
\eqref{eq:IneqJS} and \eqref{eq:IneqJS2}
we arrive thus at the bounds
relating the transmission distance relates with other measures, 
\begin{small}
\begin{align}\label{eq:bounds}
	\frac{T(\mathcal{E},\mathcal{F})}{2\sqrt{2}}\leq d_{\textrm{t}}^{\text{iso}}(\mathcal{E},\mathcal{F})\leq \min \left\{\sqrt{ T(\mathcal{E},\mathcal{F})} \ , \ D_E(\mathcal{E},\mathcal{F})\right\}. 
\end{align}
\end{small}
Further discussion of the 
upper bound 
%  \eqref{eq:tracedistanceChoi}
is provided in Appendix \ref{sec:appendixBounds}. 

\section{Entropic channel divergence\label{sec:entropicdisting}}

\noindent
 Let us now explore another approach to introduce
 a distinguishability measure into the space of maps
 by using the transmission distance.
 %described in Sec. \ref{sec:prelim}.
 %As it  pointed out in Sec. \ref{sec:Intro}, 
 The quantum Jensen-Shannon divergence plays a key role in quantum information theory as the maximal amount of classical information transmissible by means of quantum ensembles \cite{Watrous2009}.
 For  a given  quantum channel $\mathcal{E}$ 
 one defines its Holevo capacity,
\begin{align*}
	C_1(\mathcal{E})=\max_{\Pi}\text{QJSD}_{\vect{p}}[\mathcal{E}(\rho_1),\dots,\mathcal{E}(\rho_n)],
\end{align*}
where the maximum is taken over all ensembles
 $\Pi=\{p_i,\rho_i\}_{i=1}^n$.

Consider now a different setup, in which 
a fixed state $\rho$ is transformed by channel 
$\mathcal{E}_i$ with probability $p_i$.
The associated Holevo information  \cite{Holevo1973} reads %\cite{Ho72,Ho73}
\begin{eqnarray}\label{eq:Holevorho}
\mathcal{X}(\rho)=\text{QJSD}_{\vect{p}}[\mathcal{E}_1(\rho),\dots,\mathcal{E}_n(\rho)]. 
\end{eqnarray}
Taking two analyzed channels 
$\mathcal{E}_1$ and $\mathcal{E}_2$
with equal weights, $p_1=p_2=1/2$, we arrive at a \textit{worst-case distance} measure 
between them,
\begin{align}\label{eq:entropicchanneldiv}
	d_{\textrm{t}}(\mathcal{E},\mathcal{F})=\sup_{\rho\in\mathcal{M}_N}\sqrt{\text{QJSD}[\mathcal{E}(\rho),\mathcal{F}(\rho)]}.
\end{align}
Without loss of generality 
the supremum can be restricted 
to pure states \cite{Wilde2020a}.

In the above definition one analyses 
directly the action of the channels 
$\mathcal{E}_i$ on the state $\rho$ of  size $N$.
A more general approach involves
extending the system by a $K$-dimensional ancilla
 \cite{Gilchrist2005,Leditzky2018a}
and studying the action of extended
channels, $\mathcal{E}_i \otimes \mathbbm{1}_K$.
The 
 \textit{entropic channel divergence} reads
%\begin{small}
\diego{\begin{align}\label{eq:entropicchanneldivstab}
	d_{\text{t}}^{K}\!(\mathcal{E},\mathcal{F})\!=\!
	\sup_{\sigma\in\mathcal{M}_{N K}}\!d_{\text{t}}[(\mathcal{E}\!\otimes\!\mathbbm{1}_K)(\sigma),(\mathcal{F}\!\otimes\!\mathbbm{1}_K)(\sigma)],
	\end{align}
}

\noindent where the state $\sigma$  acts on an extended space of size $NK$. Observe that in the special case $K=1$ one has $d_{\text{t}}^{K=1}(\mathcal{E},\mathcal{F})=d_{\textrm{t}}(\mathcal{E},\mathcal{F})$, as expected.

\subsection{Properties of $d_{\text{t}}^{K}(\mathcal{E},\mathcal{F})$ and  the chain rule}\label{sec:chainrule}

\noindent Let us discuss some key properties of the entropic channel divergence.  By definition,
for an arbitrary dimension $K$ of the ancilla, 
 the entropic channel divergence $d_{\text{t}}^{K}$  is symmetric, null if and only if the maps are equal, and satisfies the triangle inequality in the space of quantum channels. On the other hand, we have,
\begin{align}
&d_{\text{t}}^{K}[\mathcal{E}\otimes\mathbbm{1}_K,\mathcal{F}\otimes\mathbbm{1}_K]\geq d_{\text{t}}[(\mathcal{E}\otimes\mathbbm{1}_K)(\rho^*),(\mathcal{F}\otimes\mathbbm{1}_K)(\rho^*)] \nonumber \\
&\geq d_{\text{t}}[\text{Tr}_K \ (\mathcal{E}\otimes\mathbbm{1}_K)(\rho^*),\text{Tr}_K \ (\mathcal{F}\otimes\mathbbm{1}_K)(\rho^*)] \nonumber\\ \label{eq:nonstab}
&=d_{\text{t}}[\mathcal{E}(\rho_Q^*),\mathcal{F}(\rho_Q^*)]= d_{\text{t}}(\mathcal{E},\mathcal{F}),
\end{align}
where $\rho^*_Q$ denotes the state which maximizes $d_{\text{t}}(\mathcal{E},\mathcal{F})$, while
$\rho^*$ is any joint density matrix in $\mathcal{M}_{N  K}$ such that $\TrP{K}{\rho^*}=\rho^*_Q$.
In the same way, 
for any $K'$ being a multiple of $K$,
it is possible to show the following relation,
$$d_{\text{t}}^{K}(\mathcal{E},\mathcal{F})\leq d_{\text{t}}^{K'}(\mathcal{E},\mathcal{F}).$$
 This inequality  suggests that $d_{\text{t}}^{K}$ is in general 
not stable under the addition of an ancillary systems. 
Furthermore, it was shown in \cite{Aharonov1998} 
that if $K<N$ the channel divergence arising from the trace norm is in general not stable with respect to tensor product. To ensure stability one supplies the requirement 
that the size of the ancilla and the principal systems are equal, $K=N$.
It was demonstrated in \cite{Gilchrist2005}
that for $K\geq N$ the following equality holds: $$d_{\text{t}}^{K}(\mathcal{E},\mathcal{F})=d_{\text{t}}^{N}(\mathcal{E},\mathcal{F}).$$
 This implies that  for $K=N$
 the entropic channel divergence is stable under 
 the addition of auxiliary subsystems,	 $$d_{\text{t}}^{N}(\mathcal{E},\mathcal{F})=d_{\text{t}}^{N}(\mathcal{E}\otimes \mathbbm{1},\mathcal{F}\otimes \mathbbm{1}).$$  

As a result, it is natural to choose $K=N$
and in this work the quantity 
$d_{\text{t}}^{N}$ will be called
\textit{stabilized} entropic channel divergence.

The chaining property, post-processing inequality and unitary invariance can be straightforwardly demonstrated by using the monotonicity and triangle inequality of the transmission distance in the state space \cite{Gilchrist2005}.  

Once defined $d_{\text{t}}^{K}(\mathcal{E},\mathcal{F})$, we can establish a \textit{chain rule} for the entropic channel divergence, analogously to that obtained for the quantum relative entropy in Ref. \cite{Fang2020} -- this should not be confused with the chaining property discussed above. 
\begin{proposition}
Let $\mathcal{E}$ and $\mathcal{F}$ denote arbitrary two operations acting over $\mathcal{M}_N$. For arbitrary bi-partite quantum states $\rho$ and $\sigma$ in $\mathcal{M}_{N  K}$ the following chain rule holds, 
\begin{align}\label{eq:chainrule}
d_{\text{t}}\big[(\mathcal{E}\!\otimes\!\mathbbm{1}_K)\!(\rho),(\mathcal{F}\!\otimes\!\mathbbm{1}_K)\!(\sigma)\big]\leq d_{\textrm{t}}(\rho,\sigma)\!+\! d_{\text{t}}^{K}\!(\mathcal{E},\mathcal{F}). 
\end{align}
\end{proposition}
\noindent 

It relates the transmission distance $d_{\text{t}}(\cdot,\cdot)$
between quantum states, defined in \eqref{eq:tranmissiondistance}, and
the entropic channel divergence $d_{\text{t}}^{K}(\cdot,\cdot)$ 
introduced in Eq. \eqref{eq:entropicchanneldivstab}.

\begin{proof}
It will be convenient to use a simpler notation and write $\mathcal{E}_{NK}(\rho)$ instead of $(\mathcal{E}\otimes\mathbbm{1}_K)(\rho)$ for a quantum operation $\mathcal{E}$ acting on $\mathcal{M}_N$. 
Using this convention, we have,
	\begin{align}
	d_{\text{t}}[\mathcal{E}_{NK}(\rho),\mathcal{F}_{NK}(\sigma)]&\leq d_{\text{t}}[\mathcal{E}_{NK}(\rho),\mathcal{E}_{NK}(\sigma)]+\nonumber\\
	&+d_{\text{t}}[\mathcal{E}_{NK}(\sigma),\mathcal{F}_{NK}(\sigma)] \nonumber \\&\leq d_{\textrm{t}}(\rho,\sigma)
	+d_{\text{t}}[\mathcal{E}_{NK}(\sigma),\mathcal{F}_{NK}(\sigma)] \nonumber\\&\leq d_{\textrm{t}}(\rho,\sigma)+ d_{\text{t}}^{K}(\mathcal{E},\mathcal{F}),
	\end{align}
in which we have employed the triangle inequality and the monotonicity of the transmission distance.
\end{proof}

Note that  the chain rule \eqref{eq:chainrule} is valid not only for the stabilized version of the entropic channel divergence but also for the original version \eqref{eq:entropicchanneldiv} and the maps applied directly over the states describing the 
principal $N$-dimensional system.

The chain rule \eqref{eq:chainrule} has interesting applications in the context of hypothesis testing in quantum channel discrimination \cite{Fang2020}, due to its connection with the \textit{amortized channel divergence}, introduced  in \cite{Wilde2020a} for an arbitrary generalized divergence $d(\cdot,\cdot)$. By using the transmission distance, we obtain the amortized entropic divergence, 
\begin{small}
\begin{align}
\label{eq:amort}
d_{\text{t}}^A(\mathcal{E},\mathcal{F})=\sup_{\rho, \sigma \in \mathcal{M}_{NK}}\!\left\{ d_{\text{t}}[\mathcal{E}_{NK}(\rho),\mathcal{F}_{NK}(\sigma)]-d_{\textrm{t}}(\rho,\sigma) \right\},
\end{align}
\end{small}

\noindent which depends on the size $K$ of the ancilla. Note that the chain rule \eqref{eq:chainrule} 
establishes an upper bound for $d_{\text{t}}^A(\mathcal{E},\mathcal{F})$.
A lower bound,  $$d_{\text{t}}^A(\mathcal{E},\mathcal{F})\geq d_{\text{t}}^{K}(\mathcal{E},\mathcal{F}),$$
was shown  \cite{Wilde2020a} to hold 
for an arbitrary distance measures $d(\cdot,\cdot)$. 
We arrive therefore at the \textit{amortization collapse} of the entropic channel divergence, that is 
\begin{align}
\label{eq:amortcoll}
d_{\text{t}}^A(\mathcal{E},\mathcal{F})=d_{\text{t}}^{K}(\mathcal{E},\mathcal{F}).
\end{align}

%%%%%%%%%%%%%
\section{Physical interpretation}\label{sec:PhysInt}

\noindent We defined the transmission distance 
\eqref{eq:ChoiDisting} between quantum channels,
and the entropic channel divergence \eqref{eq:entropicchanneldivstab}
and will now discuss their physical meaning.

\subsection{Transmission distance between quantum channels}
\noindent The transmission distance between quantum channels is \textit{easy}  to compute, as its definition  does not require any optimization procedure. The calculations are reduced to 
evaluation of the entropy of a map \cite{Roga2011}, equal to the von Neumann entropy of the corresponding Choi states. 
Furthermore, it is possible to estimate experimentally this quantity,
since its definition involves the Choi states, which can be obtained 
by quantum process tomography \cite{Gilchrist2005}. 

Observe that $\left[d_{\textrm{t}}^{\text{iso}}(\mathcal{E},\mathcal{F})\right]^2$ is the Holevo information corresponding to an equiprobable ensemble composed by the states $\rho_\mathcal{E}$ and $\rho_\mathcal{F}$. Additionally, for general discrete ensembles, $\left[d_{\textrm{t}}^{\text{iso}}(\mathcal{E},\mathcal{F})\right]^2$ is connected to the protocol of \textit{dense coding}. 
Consider a bipartite quantum system in a maximally entangled
state, $\rho_{r}=|\Phi\rangle \langle \Phi|$, 
usually known as \textit{resource state}, subjected to local unitary transformations $\mathcal{U}_i$ performed with probability $p_i$. 
%In the noiseless superdense coding protocol, it is employed a maximally entangled state $\ket{\Phi}$ as a resource (see Eq. \eqref{eq:Bellstate}): $\rho_r=\ket{\Phi}\bra{\Phi}$. 
The output state
	\begin{align}
	\rho_i^U=(U_i\otimes \mathbbm{1})\rho_r(U_i^\dagger\otimes \mathbbm{1}),	
	\end{align}
occurs with probability $p_i$. This protocol, relying on the initial entanglement between both parties, 
%usually referred to as Alice and Bob, 
allows them to transmit classical information encoded in a bipartite system, while conducting operations on a single subsystem only. 
If the dimension of each subsystem is $N$, it is possible to send $2\log_2 N$ bits of classical information,
even though the classical coding allows one to send only
$\log_2 N$ bits.

The capacity of the dense coding protocol with resource $\rho_r$
 to transmit classical information for fixed unitary operations $U_i$,
is given  \cite{Laurenza2020}  by the maximum over $\{p_i\}_i$ of $\text{QJSD}_{\vect{p}}(\rho^U_1,\dots,\rho^U_n)$.
Since $\rho_i^U$ form  Choi matrices of  unitary channels, $\mathcal{U}_i$, the divergence $\left[d_{\textrm{t}}^{\text{iso}}(\mathcal{U}_1,\mathcal{U}_2)\right]^2$, 
coincides with the capacity of the coding with equal 
probabilities of all unitary operations, $p_i=1/n$.

Therefore,  $\left[d_{\textrm{t}}^{\text{iso}}(\mathcal{E},\mathcal{F})\right]^2$ is the dense coding capacity connected to maps $\mathcal{E}$ and $\mathcal{F}$, for a noiseless protocol with a maximally entangled
resource state $\rho_r$. 
Distinguishability of quantum maps using quantum dense coding protocol
was advocated by Raginsky \cite{Raginsky2001}, who analyzed an analogous measure based on the quantum fidelity instead of the quantum Jensen-Shannon divergence.

\subsection{Entropic channel divergence}

\noindent Given a collection of quantum operations $\{\mathcal{E}_i\}$ with probabilities $\{p_i\}$, the quantity 
\begin{eqnarray*}
	\sup_{\rho\in\mathcal{M}_{N\times N}}\textrm{QJSD}_{\vect{p}}[(\mathcal{E}_1\otimes \mathbbm{1})(\rho),\dots,(\mathcal{E}_n\otimes\mathbbm{1})(\rho)]
\end{eqnarray*}
is called  the {\sl quantum reading capacity}, defined  
in a scheme of readout of quantum memories \cite{Pirandola2011}. 
This process corresponds to channel decoding when a decoder retrieves information in the cells of a memory. The entropic channel divergence is the square root of the previous quantity in the symmetric case, $p_i=1/n$.

The one-shot capacity of a dense coding protocol, with an arbitrary resource state $\rho_r$, can be rewritten in terms of the quantum reading capacity \cite{Laurenza2020}.

\subsection{Relation between the channel divergence
$d_\text{t}^N$ and the transmission  distance $d_\text{t}^\text{iso}$}
\label{sec:connection}

\noindent Assume that the single-qubit channels we wish to distinguish are covariant with respect to Pauli operators. This means that for each quantum channel $\mathcal{E}$ we can write $\mathcal{E}\!\circ\! \mathcal{P} = \mathcal{P}'\!\circ\! \mathcal{E}$,
where
$\mathcal{P}$ and $\mathcal{P}'$ denote Pauli channels. In this case, the channel can be simulated with LOCC operations \cite{Laurenza2020}, and it is called \textit{Choi-stretchable},
so that
\begin{align}\label{eq:paulicov}
	\mathcal{E}(\rho)=\mathcal{T}_{\text{tele}}(\rho \otimes \rho_{\mathcal{E}}).
\end{align}
Here $\mathcal{T}_{\text{tele}}$ denotes the standard quantum teleportation protocol and $\rho_\mathcal{E}$ stands for 
the corresponding Choi state of the map $\mathcal{E}$. Thus, for any two Choi-stretchable quantum operations $\mathcal{E}_1$ and $\mathcal{E}_2$, we have
\begin{align*}
	d_{\text{t}}^{N}(\mathcal{E}_1,\mathcal{E}_2) &=\sup_{\rho\in\mathcal{M}_{N \times N}}d_{\text{t}}[(\mathcal{E}_1\otimes\mathbbm{1})(\rho),(\mathcal{E}_2\otimes\mathbbm{1})(\rho)] \nonumber\\
	&=\sup_{\rho\in\mathcal{M}_{N \times N}} d_{\text{t}}[\mathcal{T}_{\text{tele}}(\rho \otimes \rho_{\mathcal{E}_1}),\mathcal{T}_{\text{tele}}(\rho \otimes \rho_{\mathcal{E}_2})] \nonumber \\
	&\leq \sup_{\rho\in\mathcal{M}_{N \times N}} d_{\text{t}}(\rho \otimes \rho_{\mathcal{E}_1},\rho \otimes \rho_{\mathcal{E}_2})=d_{\text{t}}^{\text{iso}}( \mathcal{E}_1, \mathcal{E}_2).
\end{align*}
We applied here the sub-additivity of the QJSD in the state space and its monotonicity under CP maps. By definition of $d_{\text{t}}^{N}$, 
inequality holds $d_{\text{t}}^{\text{iso}}( \mathcal{E}_1, \mathcal{E}_2)\leq d_{\text{t}}^{N}(\mathcal{E}_1,\mathcal{E}_2)$. 
Thus the equality
\begin{align}
d_{\text{t}}^{\text{iso}}( \mathcal{E}_1, \mathcal{E}_2)=d_{\text{t}}^{N}(\mathcal{E}_1,\mathcal{E}_2)
\end{align}
is valid for any two  Pauli covariant operations $\mathcal{E}_1$ and $\mathcal{E}_2$.

\section{Applications} \label{sec:applications}

\noindent In this section, we explore certain features of the distinguishability measures between quantum operations proposed in Sections \ref{sec:transmissiondistancebetQC} and \ref{sec:entropicdisting}. We analyze two particular
single-qubit problems: distinguishing two unitary Pauli operations 
and two Hamiltonian evolutions under decoherence. 
%Before beginning, it is helpful to define shortly the \textit{affine decomposition}.

The three-dimensional Bloch vector $\vect{r}$ of a single-qubit state
allows us to represent the density matrix as
\begin{align}
	\rho=\frac{1}{2}\left(\mathbbm{1}+\vect{r}\cdot\vect{\sigma}\right).
\end{align}
Here $\vect{r}\cdot\vect{\sigma}=\sum_{i=1}^3 r_i \sigma_i$ with $\{\sigma_i\}_i$ denoting three Pauli matrices. The action of a quantum operation $\mathcal{E}$ over $\rho$ can be described by
a distortion matrix $\Lambda_{\mathcal{E}}$ and a translation vector $\vect{l}_{\mathcal{E}}$,
\begin{align}
	\mathcal{E}(\rho)&=\frac{1}{2}\left(\mathbbm{1}+\vect{r}_{\mathcal{E}}\cdot 
	\vect{\sigma}\right) \ \text{with} \nonumber \\
	\vect{r}_{\mathcal{E}}&=\Lambda_{\mathcal{E}} \vect{r}+\vect{l}_{\mathcal{E}}. \label{eq:affinedecomp}
\end{align}
The above form is called the affine decomposition or 
the Fano representation of the map.
%while $\Lambda_{\mathcal{E}}$ and $\vect{l}_{\mathcal{E}}$ are referred to as distortion matrix and translation vector corresponding to the map $\mathcal{E}$, respectively.

\subsection{Pauli channels}\label{sec:PauliChann}

\noindent All single-qubit unital operations belong to the class
of Pauli channels,
%This kind of unital operations can be written as
\begin{align}
\label{eq:Paulimap}
	\mathcal{P}_p(\rho)=\sum_{\alpha=0}^3 p_\alpha \sigma_\alpha \rho \sigma_\alpha,
\end{align}
where $\{\sigma_\alpha\}_{\alpha=0}^3=\{\mathbbm{1},\vect{\sigma}\}$ and $\{p_\alpha\}_{\alpha=0}^3$ is a discrete probability vector.
The Fano form of such a map $\mathcal{P}$ reads,
\begin{align}
	\vect{l}_{\mathcal{P}}&=\vect{0},\nonumber\\
	\Lambda_{\mathcal{P}}&=\text{diag}(c_1,c_2,c_3)=\sum_{\alpha=0}^3 p_\alpha R_\alpha, \label{eq:affinepuali}
\end{align}
with
\begin{align}
	R_0&=\text{diag}(1,1,1), \nonumber\\
	R_1&=\text{diag}(1,-1,-1),\nonumber \\
	R_2&=\text{diag}(-1,1,-1), \nonumber \\
	R_3&=\text{diag}(-1,-1,1).
\end{align}
Thus, $R_\alpha$ is a diagonal orthogonal matrix defined by the action of the unitary transformations given by the Pauli matrix $\sigma_\alpha$ and $R_0$ is connected to the identity map. Additionally, the set $\vect{c}=(c_1,c_2,c_3)$, in Eq. \eqref{eq:affinepuali}, for which $\mathcal{P}$ is a well-defined CPTP map specifies a tetrahedron in the three-dimensional space 
\cite{Ruskai2002},
with edges $\{R_\alpha\}_{\alpha=0}^3$, see Fig. \ref{fig:QJSDBallChoiDepol}. The relation among $\{p_\alpha\}_{\alpha=0}^3$ and the numbers $\{c_i\}_{i=1}^3$ is
\begin{align}
	p_0&=\frac{1}{4}(1+c_1+c_2+c_3), \nonumber %\label{eq:probEigenChoiBD1}
	\\
	p_1&=\frac{1}{4}(1+c_1-c_2-c_3), \nonumber \\
	p_2&=\frac{1}{4}(1-c_1+c_2-c_3), \nonumber \\
	p_3&=\frac{1}{4}(1-c_1-c_2+c_3). \label{eq:probEigenChoiBD4}
\end{align}
Particular examples of Pauli maps are the identity, the phase flip channel $\mathcal{P}_{pf}$ and the depolarizing map $\mathcal{D}$, 
corresponding to the distortion matrices
\begin{align}
\Lambda_{\mathcal{I}}&=\text{diag}(1,1,1),\label{eq:identitymap}\\
\Lambda_{\mathcal{P}_{pf}}&=\text{diag}(1-x,1-x,1), \label{eq:phaseflip}\\
\Lambda_{\mathcal{D}}&=\text{diag}(1-x,1-x,1-x),\label{eq:depol}
\end{align}
respectively. 
Completely depolarizing channel, $\mathcal{D}_0$,
corresponds to Eq. \eqref{eq:depol}
with $x=1$.

%{\color{red} Editorial: perhaps $x=1$,
%or \\
%$\Lambda_{\mathcal{D}} =\text{diag}(x,x,x)$, 
%   Please check!} 
%\diego{Thanks for noting this Karol. It was a typo. I took $x=1$.}

For an arbitrary channel $\mathcal{E}$, the distortion matrix $\Lambda_\mathcal{E}$, can be diagonalized by applying local unitary transformations on $\mathcal{E}(\rho)$, reaching  the canonical form of the map, which is subsequently given by the translation vector $\vect{t}_\mathcal{E}=(t_1,t_2,t_3)$ and the distortion vector $\vect{\omega}_\mathcal{E}=(\omega_1,\omega_2,\omega_3)$,
which results from the diagonalization of $\Lambda_\mathcal{E}$ \cite{Luo2008a,Bengtsson}. Note that the canonical form of a given unital map, $\vect{t}_\mathcal{E}=\vect{0}$, gives a Pauli channel
\eqref{eq:Paulimap}.

The Choi matrix \eqref{eq:Choistate} of any single qubit channel in its canonical form reads \cite{Bengtsson},
\begin{small}
    \begin{align*}
\!\!\rho_\mathcal{E}= \frac{1}{4} \!
	\begin{bmatrix}
		{1+\omega_3+t_3}  && 0 &&  {t_1+i\omega_2} && {\omega_1+\omega_2}\\
		0 && \!\!{1-\omega_3+t_3}&& {\omega_1-\omega_2} && 
		{t_1+i\omega_2} \\
		{t_1-i\omega_2} && {\omega_1-\omega_2} && \!\!{1-\omega_3-t_3} && 0\\ \omega_1+\omega_2 && t_1-i\omega_2 && 0 && \!\!{1+\omega_3-t_3} \\
	\end{bmatrix}.
\end{align*}
\end{small}

\noindent If  $\vect{t}_\mathcal{E}=\vect{0}$, $\rho_\mathcal{E}$ forms a Bell-diagonal state (i.e. its eigenvectors are the four Bell states) and its eigenvalues are given by the probabilities $p_\alpha$ appearing in %Eqs. \eqref{eq:probEigenChoiBD1}-
\eqref{eq:probEigenChoiBD4}. Let us analyze the transmission distance between maps, Eq. \eqref{eq:ChoiDisting}, and the entropic channel divergence, Eq. \eqref{eq:entropicchanneldivstab}, for $K=1$ and $K=2$ (stabilized version).
\subsubsection{Transmission distance between Pauli Channels}

\begin{figure}
	\centering
	\includegraphics[width=.35\textheight]{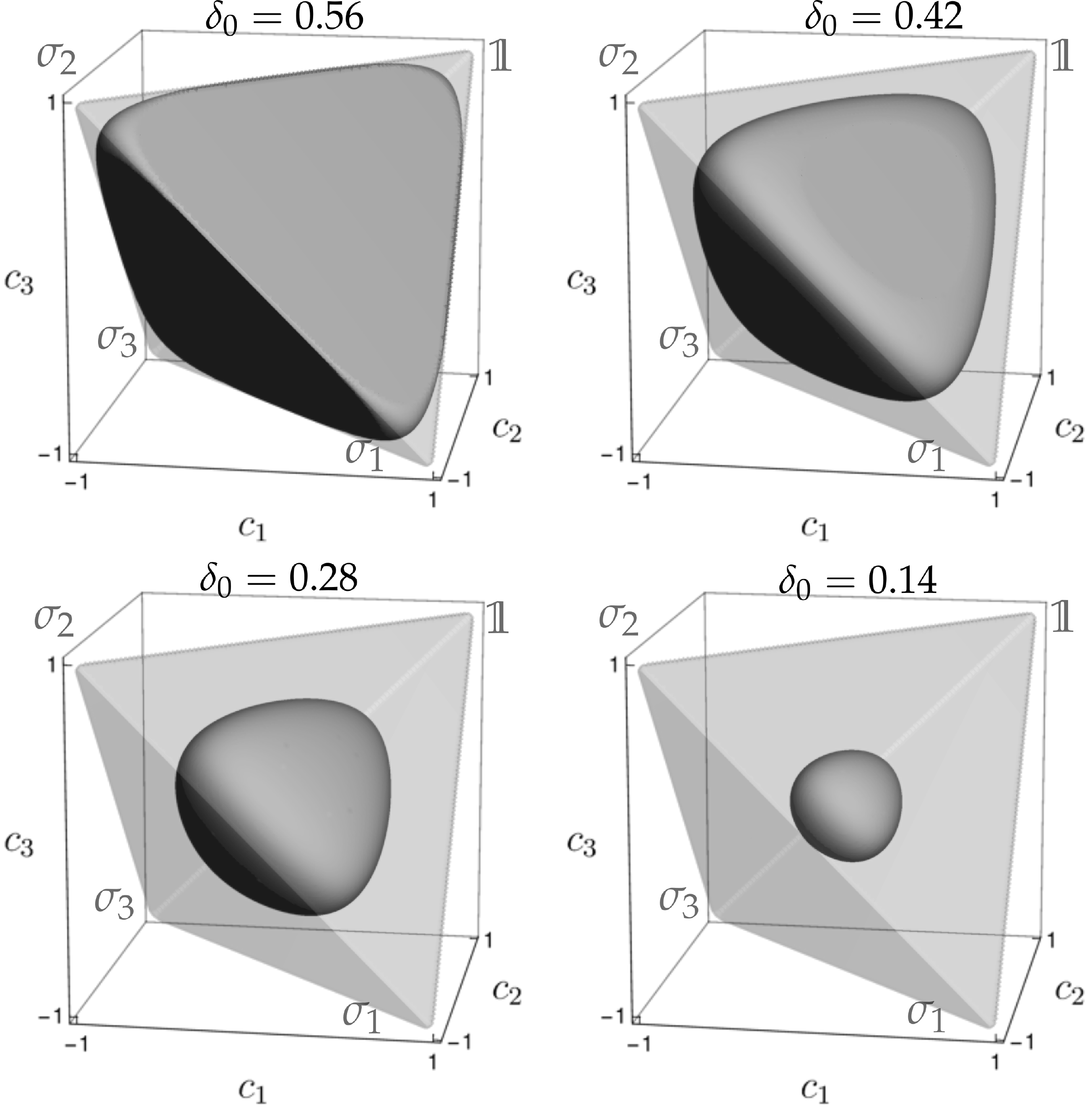}
	\caption{The tetrahedron of Pauli channels with 'spheres' of channels equidistant to the completely depolarizing channel $\mathcal{D}_0$ in the center of the tetrahedron,
	%given by \eqref{eq:depol} with $x=0$, 
	with respect to the distance $d_\text{t}^{\text{iso}}(\mathcal{P}_p, \mathcal{D}_0)=\delta_0$, for radii $\delta_0 \in \{0.56, 0.42,0.28,0.14\}$.}
	\label{fig:QJSDBallChoiDepol}
\end{figure}

\begin{figure}
	\centering
	\includegraphics[width=.35\textheight]{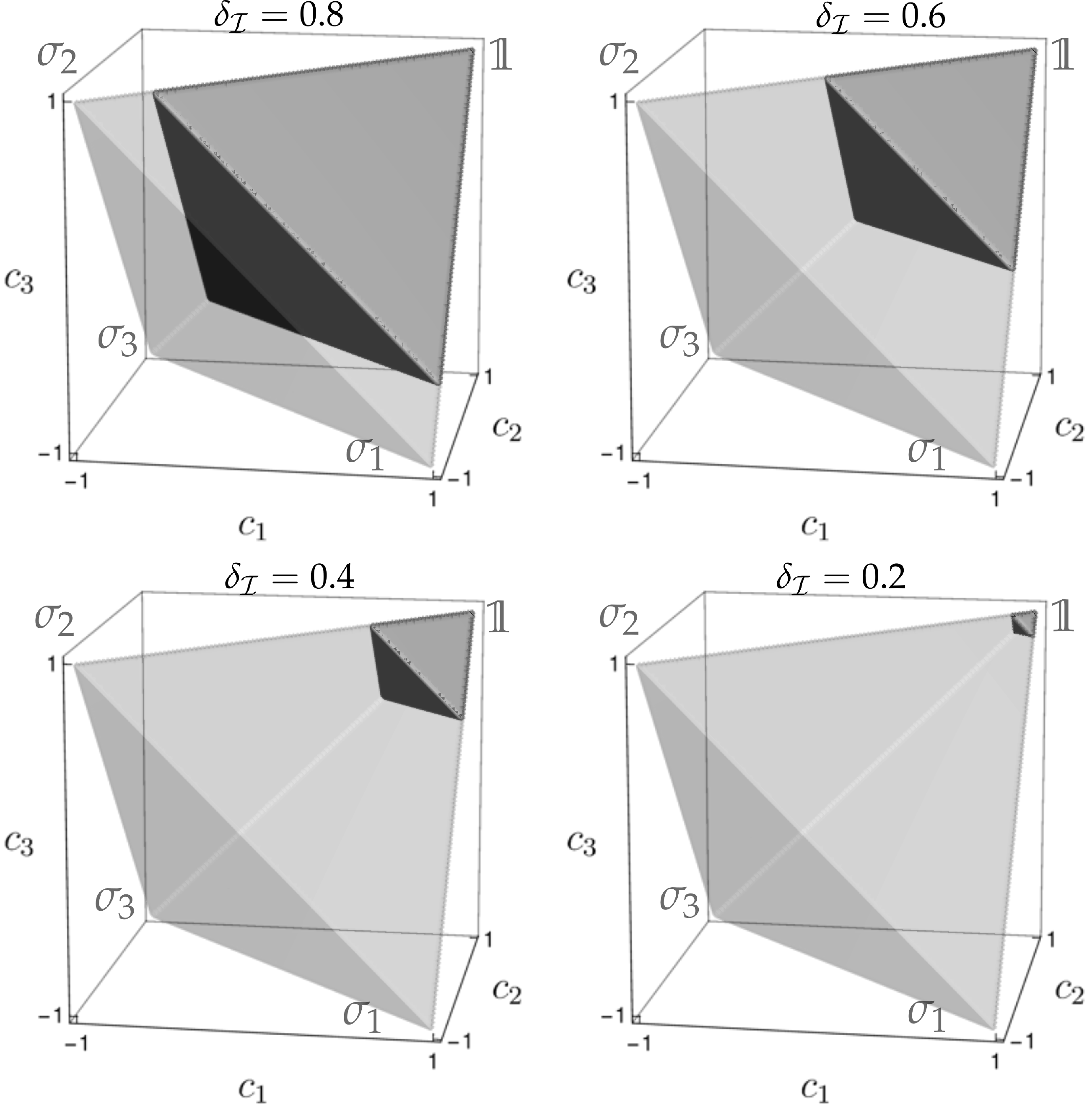}
	\caption{Surfaces within the Pauli tetrahedron, 
	 defined by a constant transmission distance to the identity map $\mathcal{I}$ represented by the corner of the set, 
	 $d_\text{t}^{\text{iso}}(\mathcal{P}_p, \mathcal{I})=\delta_\mathcal{I}$, for $\delta_\mathcal{I} \in \{0.8, 0.6,0.4,0.2\}$.}
	\label{fig:QJSDBallChoiIdent}
\end{figure}

\noindent Let $\mathcal{P}_p$ and $\mathcal{P}_q$ be two Pauli channels defined by two probability distributions $\{p_\alpha\}_{\alpha=0}^3$ and $\{q_\beta\}_{\beta=0}^3$, as in 
Eq.~\eqref{eq:Paulimap}. The corresponding Choi matrices of these maps become diagonal in the Bell basis. The quantum Jensen-Shannon divergence between two Pauli channels is therefore equal to the classical Jensen-Shannon divergence evaluated in classical tetrahedron of four-point probability distributions,
determined by the spectra of both Choi states, $p=\{p_\alpha\}_{\alpha=0}^3$ and $q=\{q_\beta\}_{\beta=0}^3$, 
\begin{align}
	d_\text{t}^{\text{iso}}(\mathcal{P}_p, \mathcal{P}_q)=\sqrt{\text{JSD}(\vect{p}||\vect{q})}.
\end{align}

Using the three-dimensional parameterization in  %\eqref{eq:probEigenChoiBD1}-
\eqref{eq:probEigenChoiBD4}, we can plot the surface, within the Pauli tetrahedron, defined by those maps with the same transmission distance to the centre of the tetrahedron, which represents the completely depolarizing map $\mathcal{D}_0$,
\begin{align}
d_\text{t}^{\text{iso}}(\mathcal{P}_p, \mathcal{D}_0)=\delta_0.
\end{align}
In Fig. \ref{fig:QJSDBallChoiDepol}, such 'spheres' 
with respect to this distance are plotted for four different radii.
%, $\delta_0 \in \{0.56, 0.42,0.28,0.14\}$.
For a small radius $\delta_0$ such a surface
resembles a sphere, while for a larger values of 
$\delta_0$ it becomes deformed by the
faces of tetrahedron.

Analogously, Fig. 2 presents  four 'spheres' corresponding to the fixed transmission distance to the identity map, $d_\text{t}^{\text{iso}}(\mathcal{P}_p, \mathcal{I})=\delta_\mathcal{I}$, with radii $\delta_\mathcal{I}$ listed in the caption.
%being $\delta_\mathcal{I}\in \{0.8, 0.6,0.4,0.2\}$.

In Fig. \ref{fig:QJSDvsTRaceforchoiandCD}, we plot the transmission distance between the maps given by \eqref{eq:identitymap}-\eqref{eq:depol}, as functions of the depolarizing parameter $x\in [0,1]$, and the trace distance between the corresponding Choi states. For $x\not= 0$, we observe that
\begin{align}
	d_\text{t}^{\text{iso}}(\mathcal{P}_{pf}, \mathcal{I}) < d_\text{t}^{\text{iso}}(\mathcal{P}_{pf}, \mathcal{D}) < d_\text{t}^{\text{iso}}(\mathcal{I}, \mathcal{D}),
\end{align}
while for the trace distance  \eqref{eq:tracedistanceChoi} the following relations hold,
\begin{align}
T(\mathcal{P}_{pf}, \mathcal{D})= T(\mathcal{P}_{pf}, \mathcal{I}) < T(\mathcal{I}, \mathcal{D}).
\end{align}	

\begin{figure}
	\centering
	\includegraphics[width=.35\textheight]{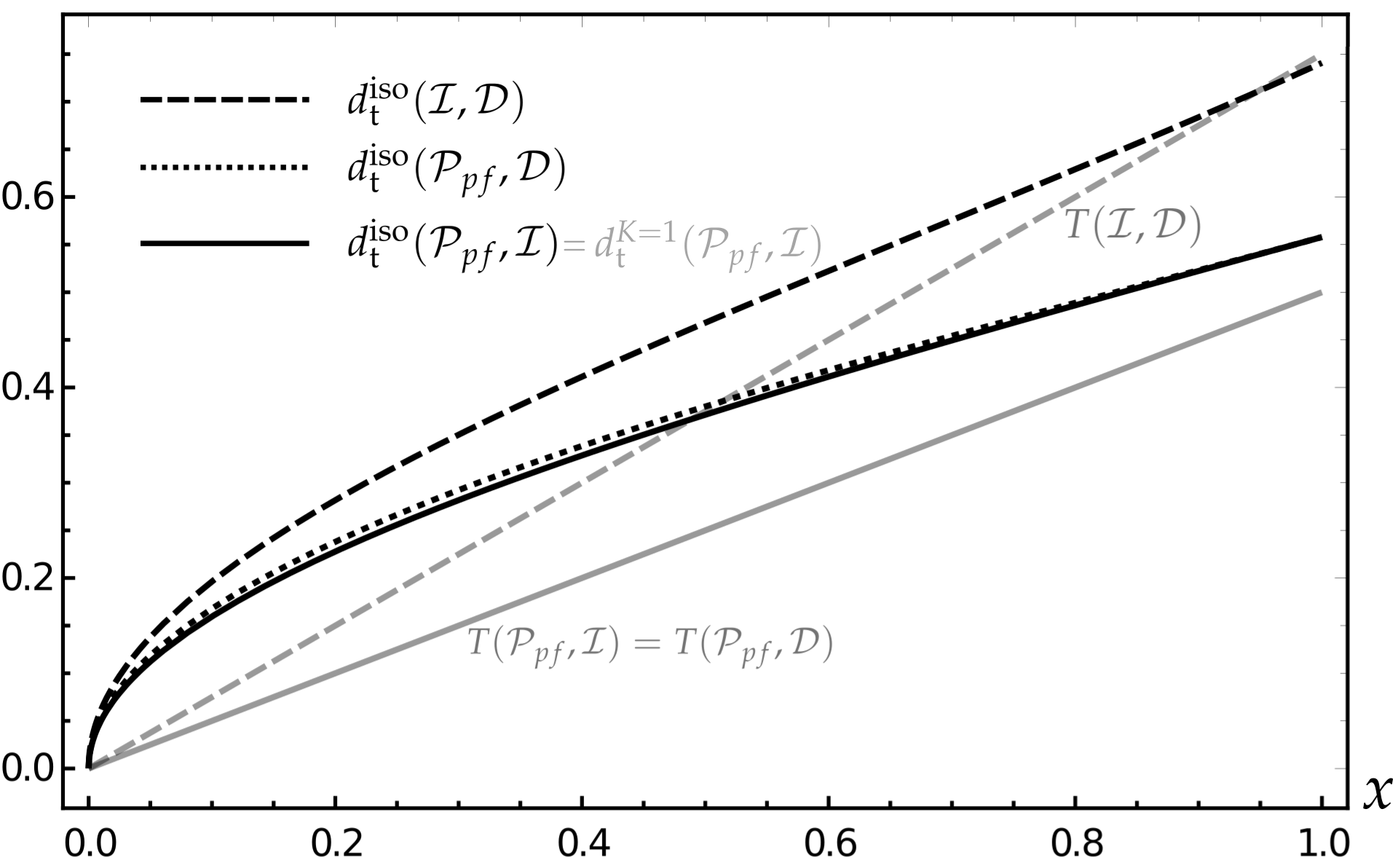}
	\caption{\textit{Phase-flip noise teleportation}: Transmission distance $d_{\text{t}}^{\text{iso}}$ defined in \eqref{eq:ChoiDisting} between the identity map, phase flip and depolarizing channel,  \eqref{eq:identitymap}-\eqref{eq:depol}, respectively, as functions of the depolarizing parameter $x$. 
	  For comparison we plot also 
	 the trace distance $T$ between 
	the corresponding Choi states,
	 see \eqref{eq:tracedistanceChoi},
	 and the entropic channel divergence $d_{\text{t}}^{K=1}$, see \eqref{eq:entropicchanneldiv}.}
	\label{fig:QJSDvsTRaceforchoiandCD}
\end{figure}

\subsubsection{Entropic channel divergence}

\noindent Let us calculate the entropic channel divergence \eqref{eq:entropicchanneldivstab} for two Pauli channels
$\mathcal{P}_p$ and $\mathcal{P}_q$ 
corresponding to probability distributions $p$ and $q$,
determined by the vectors %which are connected by Eqs. \eqref{eq:probEigenChoiBD1}-\eqref{eq:probEigenChoiBD4} to the vectors 
$\vect{c}_p=(c_{p1},c_{p2},c_{p3})$ and $\vect{c}_q=(c_{q1},c_{q2},c_{q3})$, respectively. For $N=2$, there are  two different entropic divergence measures 
labeled by the dimension $K$ of the ancilla,
$$d_{\text{t}}^{K=1}(\mathcal{E},\mathcal{F})\text{ and }d_{\text{t}}^{K=2}(\mathcal{E},\mathcal{F}),$$ 
since  $d_{\text{t}}^{K'}(\mathcal{E},\mathcal{F})=d_{\text{t}}^{K=2}(\mathcal{E},\mathcal{F})$ for $K'>2$, as mentioned before. 
The Pauli channels are Pauli covariant \eqref{eq:paulicov}, 
which implies that $d_{\text{t}}^{K=2}(\mathcal{P}_p,\mathcal{P}_q)=d_{\text{t}}^{\text{iso}}(\mathcal{P}_p,\mathcal{P}_q)$,
see Sec. \ref{sec:connection}.

In the case $K=1$, one has to optimize
the transmission distance between the channels
over the initial pure states,
\begin{align*}
d^{K=1}_{\textrm{t}}(\mathcal{P}_p,\mathcal{P}_q)=\sup_{\rho\in\mathcal{M}_N}\sqrt{\text{QJSD}[\mathcal{P}_p(\rho),\mathcal{P}_q(\rho)]},
\end{align*}
where 
\begin{small}
	\begin{align*}
	\text{QJSD}\!\left[\mathcal{P}_p(\rho),\mathcal{P}_q(\rho)\right]=\text{S}[\overline{\mathcal{P}}(\rho)]-\frac{1}{2}\text{S}[\mathcal{P}_p(\rho)]-\frac{1}{2}\text{S}[\mathcal{P}_q(\rho)],
	\end{align*}
\end{small}

\noindent and
$\overline{\mathcal{P}}=(\mathcal{P}_p+\mathcal{P}_q)/{2}$ is the average channel,  which also forms a Pauli map. 

\begin{proposition}
	Entropic channel divergence \eqref{eq:entropicchanneldiv}, between two Pauli maps $\mathcal{P}_p$ and $\mathcal{P}_q$, given by distortion matrices $\Lambda_p=(c_{p1},c_{p2},c_{p3})$ and $\Lambda_q=(c_{q1},c_{q2},c_{q3})$, 
	takes the form,
	\begin{small}
		\begin{align}
		d_{\text{t}}^{K=1}\left(\mathcal{P}_p,\mathcal{P}_q\right)=\max_{i} \sqrt{f\!\left( \overline{c}_{i}^2 \right)-\frac{1}{2}\left[f\!\left(c_{pi}^2 \right)+f\!\left(c_{qi}^2 \right)\right]}, \label{eq:ECDmax}
		\end{align}
	\end{small}
	where $\overline{c}_i=(c_{pi}+c_{qi})/2$ and
	\begin{align}\label{eq:functionf}
	f\!(x):=H_2\left(\frac{1-\sqrt{x}}{2}\right). 
	\end{align}  	
	Here $H_2(x):=-x\log_2x -(1-x)\log_2 (1-x)$
	  stands for the binary entropy function for $x\in[0,1]$.
\end{proposition}

	\begin{figure}
	\centering
	\includegraphics[width=.35\textheight]{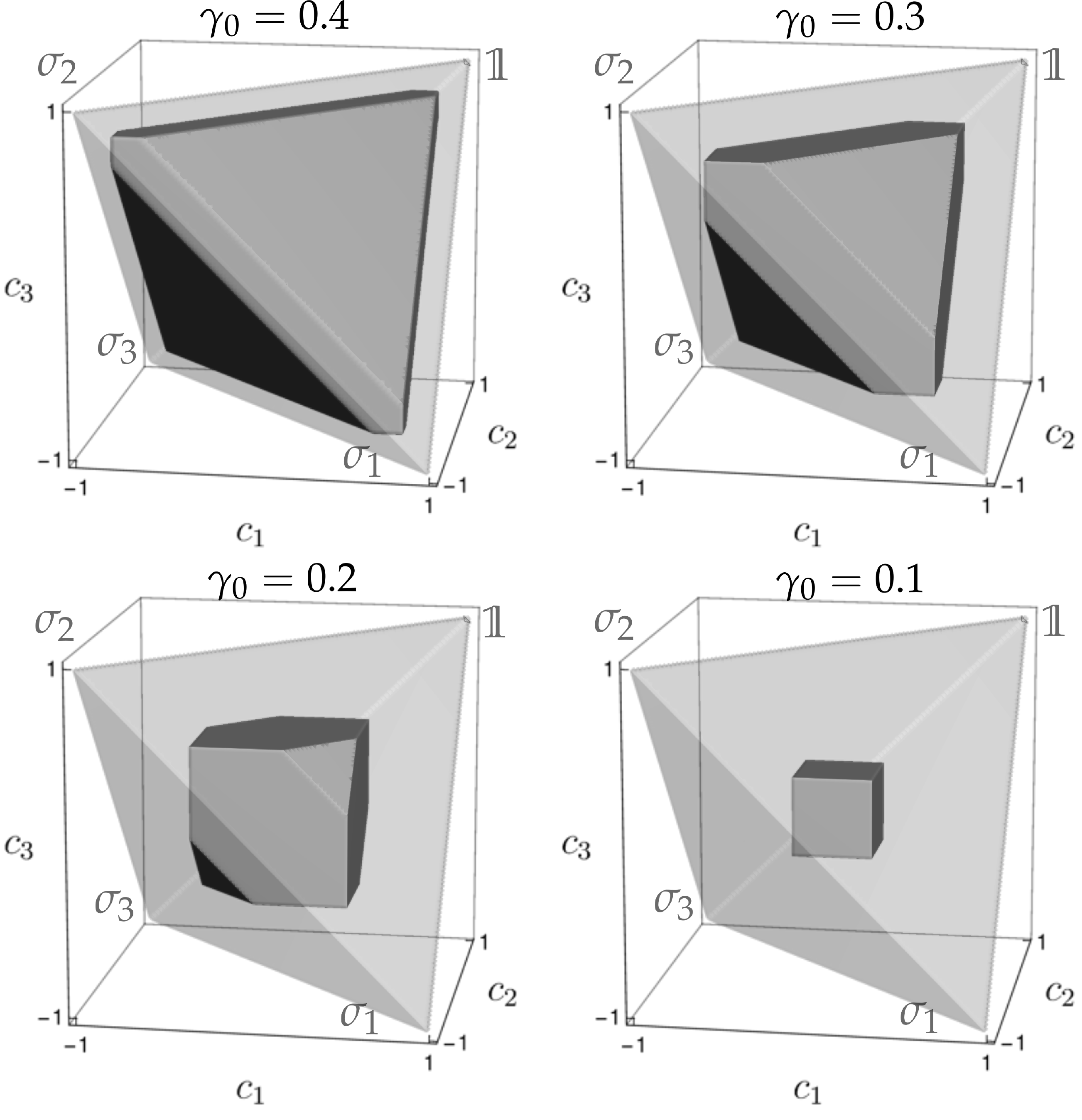}
	\caption{Spheres with respect to the distance
	$d_{\text{t}}^{K=1}$
	 within the tetrahedron of Pauli channels, $d_{\text{t}}^{K=1}\left(\mathcal{P}_p,\mathcal{D}\right)=\gamma_0$
		for four different radii: $\gamma_0 \in \{0.4, 0.3,0.2,0.1\}$.}
	\label{fig:QJSDBallECDDepol}
\end{figure}

\begin{figure}
	\centering
	\includegraphics[width=.35\textheight]{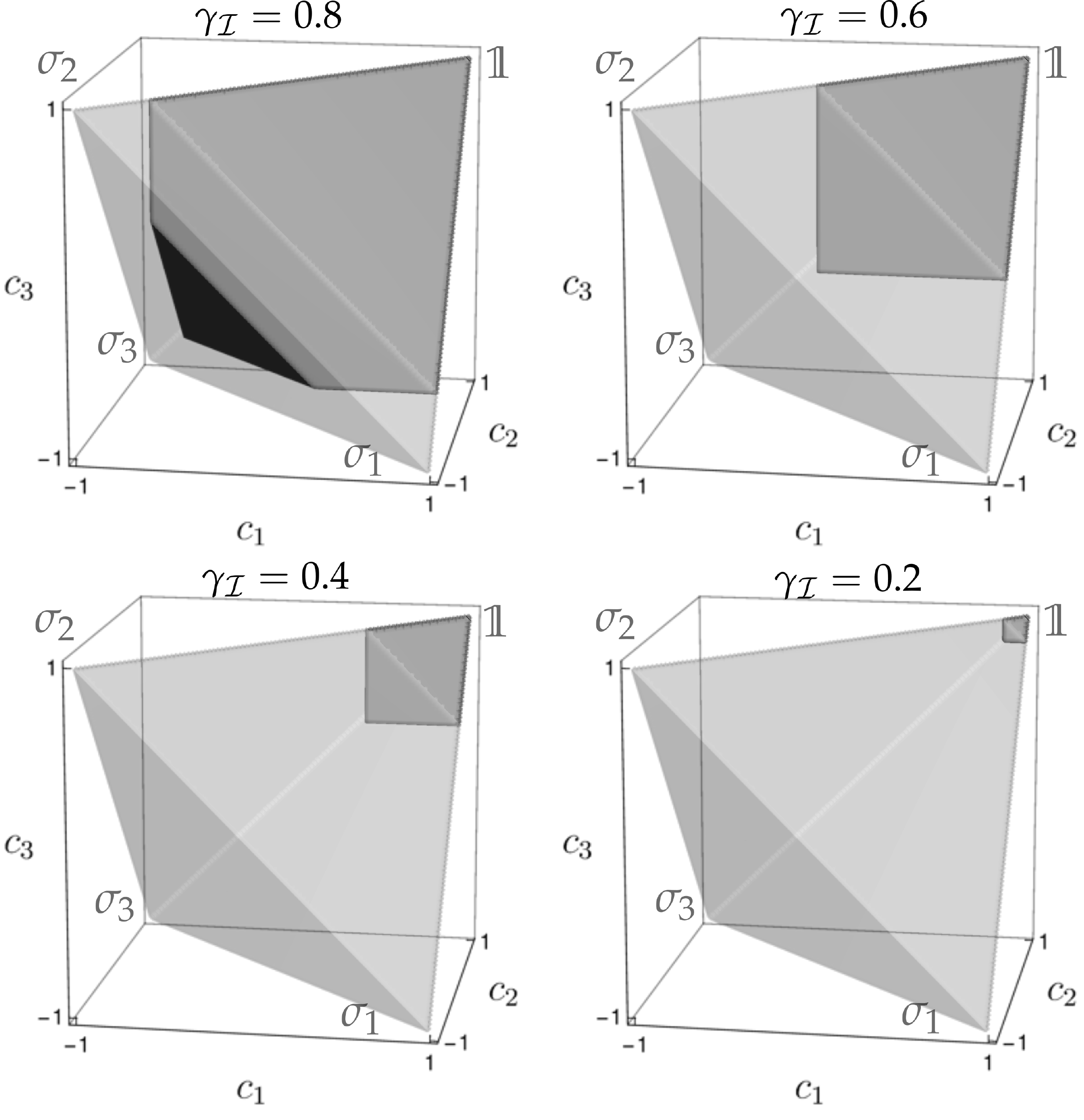}
	\caption{Surfaces defined by constant entropic channel divergence to the identity map: $d_{\text{t}}^{K=1}(\mathcal{P}_p, \mathcal{I})=\gamma_\mathcal{I}$ with $\gamma_\mathcal{I}\in \{0.8, 0.6,0.4,0.2\}$.}
	\label{fig:QJSDBallECDIdent}
\end{figure}

\begin{proof}
	For an arbitrary Pauli map $\text{S}[\mathcal{P}(\rho)]$, we have,
	\begin{align}
	\text{S}[\mathcal{P}(\rho)]=f\left(r_{\mathcal{P}}^2\right),
	\end{align}
	where $r_{\mathcal{P}}=\left|\vect{r}_{\mathcal{P}}\right|$ being $\vect{r}_{\mathcal{P}}=\Lambda_{\mathcal{P}}\vect{r}$ the Bloch vector of $\mathcal{P}(\rho)$, Eq. \eqref{eq:affinepuali}. Thus, we can write $\text{S}[\mathcal{P}(\rho)]=f\!(r_{\mathcal{P}}^2)=f\!(\vect{r}\cdot \Lambda_{\mathcal{P}}^2 \vect{r})$ and $\vect{r}^2=1$. Once we have rewritten the entropies of the Pauli channels, the quantum Jensen-Shannon divergence reads,
\begin{align}
		\text{QJSD}\!\left[\mathcal{P}_p(\rho),\mathcal{P}_q(\rho)\right]&=f\!\left(\vect{r}\cdot \Lambda_{\overline{\mathcal{P}}}^2 \vect{r} \right)-\frac{1}{2}f\!\left(\vect{r}\cdot \Lambda_{p}^2 \vect{r} \right)+\nonumber\\
		&-\frac{1}{2}f\!\left(\vect{r}\cdot \Lambda_{q}^2 \vect{r} \right).\label{eq:extremumQJSD}
\end{align}
Let us apply the method of Lagrange multipliers to the Cartesian coordinates of $\vect{r}$. This leads to the following three equations,
	\begin{small}
		\begin{align*}
		\lambda r_i=\left[f'(\vect{r}\cdot \Lambda_{\overline{\mathcal{P}}}^2 \vect{r}) \overline{c}_{i}^2-\frac{\left(f'(\vect{r}\cdot \Lambda_{p}^2 \vect{r})c_{pi}^2+f'(\vect{r}\cdot \Lambda_{q}^2 \vect{r})c_{qi}^2\right)}{2}\right]r_i,
		\end{align*}
	\end{small}

\noindent  with $i=1,2,3$,
which hold simultaneously with the constraint $\vect{r}^2=1$, associated to the Lagrange multiplier $\lambda$. Thus, 
the previous equation defines
six possible extreme values of the function \eqref{eq:extremumQJSD}
\begin{align}
	\vect{r}_{\text{opt},1}^\pm&=\pm(1, 0,0)=\pm\vect{r}_1 \label{eq:extre1}\\
	\vect{r}_{\text{opt},2}^\pm&=\pm(0,1,0)=\pm\vect{r}_2\\
	\vect{r}_{\text{opt},3}^\pm&=\pm(0, 0,1)=\pm\vect{r}_3.\label{eq:extre3}
\end{align}
As Eq. \eqref{eq:extremumQJSD} is symmetric under reflection, $\vect{r}'=-\vect{r}$, we have only three extremes that lead to different values of the QJSD. Correspondingly, the maximum is determined by Eq. \eqref{eq:ECDmax}.
\end{proof}

In Fig. \ref{fig:QJSDBallECDDepol}, we plot the three-dimensional 'spheres' within the tetrahedron of Pauli channels such that $$d_{\text{t}}^{K=1}\left(\mathcal{P}_p,\mathcal{D}_0\right)=\gamma_0$$
for four different radii.
%$\gamma_0 \in \{0.4, 0.3,0.2,0.1\}$. 
Analogously, Fig. \ref{fig:QJSDBallECDIdent} shows surfaces of maps of the same entropic channel divergence  to the identity map,
$d_{\text{t}}^{K=1}(\mathcal{P}_p, \mathcal{I})=\gamma_\mathcal{I}$,
for four exemplary values of 
$\gamma_\mathcal{I}$. %\in \{0.8, 0.6,0.4,0.2\}$.
%%%  these numbers are already given in the caption 
% and repetitions are redundant !

Consider now the distinguishability between the identity map, the phase flip and the depolarizing channel, specified in \eqref{eq:identitymap}-\eqref{eq:depol}, respectively. 
In this case, for any $x\in[0,1]$
the following inequalities hold,
\begin{align}
	d_\text{t}^{K=1}(\mathcal{P}_{pf}, \mathcal{I}) &= d_\text{t}^{K=1}(\mathcal{P}_{pf}, \mathcal{D})= d_\text{t}^{K=1}(\mathcal{I}, \mathcal{D})\nonumber\\
	&=d_\text{t}^{\text{iso}}(\mathcal{P}_{pf}, \mathcal{I}), \label{eq:entropicdistQT}
\end{align}
 Fig. \ref{fig:QJSDvsTRaceforchoiandCD} shows the dependence of this distance on the depolarizing parameter $x$.

A similar behaviour can be obtained for the distinguishability measures arising from the channel divergence based 
on the trace distance,
\begin{align}\label{eq:channeldivTr}
d_{\text{Tr}}^{K=1}(\mathcal{E},\mathcal{F})=\sup_{\rho\in\mathcal{M}_N}T[\mathcal{E}(\rho),\mathcal{F}(\rho)].
\end{align}

\begin{proposition}
Let $\mathcal{N}$ and $\mathcal{M}$ be two unital operations for $N=2$. Then,
		\begin{align}\label{eq:channeldivTrUnital}
d_{\text{Tr}}^{K=1}(\mathcal{N},\mathcal{M})=\frac{1}{2}\max_i \sqrt{\lambda_i^{\Delta}},
\end{align}
where $\{\lambda_i^{\Delta}\}_i$ is the set of eigenvalues of the matrix $$\Delta=(\Lambda_\mathcal{N}-\Lambda_\mathcal{M})^{\intercal}(\Lambda_\mathcal{N}-\Lambda_\mathcal{M}).	$$ 
\end{proposition}

A proof of this result is provided 
in Appendix \ref{sec:tracedistanceappendix}. 
The reasoning presented above implies that,
\begin{align}
	d_\text{Tr}^{K=1}(\mathcal{P}_{pf}, \mathcal{I}) &= d_\text{Tr}^{K=1}(\mathcal{P}_{pf}, \mathcal{D})= d_\text{Tr}^{K=1}(\mathcal{I}, \mathcal{D})\nonumber\\
	&=T(\mathcal{P}_{pf}, \mathcal{I}).\label{eq:tracedistanceQT}
\end{align}
Dependence of this function on the depolarizing parameter $x$ is also marked in Fig. \ref{fig:QJSDvsTRaceforchoiandCD}.

\subsubsection{Noise in quantum teleportation protocol}

\noindent Quantum teleportation, %(QT), 
one of the most important quantum information protocols, replicates the state of one quantum system into another without having information about the input state. 
This protocol requires three qubits which are operated by two different entities, usually referred to as Alice and Bob. 

The corresponding tasks to teleport the qubit state $\rho_a$ of Alice to Bob, assuming they share a two-qubit state $AB$ in the maximally entangled Bell state $\ket{\Psi}_{AB}$, are:

1) Alice measures a projection onto the Bell basis for the qubits $aA$ and classically communicates its outcome to Bob, 

2) Bob applies suitable unitary operations, according to the shared measurement result, on his qubit $B$,
to replicate the initial input state  $\rho_a$ of Alice.

Such a teleportation protocol is called \textit{perfect}
and it can be described by the identity channel,
$\mathcal{I}_{a\to B}$, with distortion matrix given by  \eqref{eq:identitymap}, where the subindex $a \to B$ denotes that the channel takes states of qubit $a$ and returns the states of qubit $B$. However, the maximally entangled state $\ket{\Psi}_{AB}$, pre-shared by Alice and Bob, can be affected by noise or decoherence. The \textit{standard} teleportation protocol consists of the above steps, but instead assuming pre-shared maximally entanglement between the qubits $AB$,
one replaces it by a \textit{resource state}, 
$$\ket{\Psi}\bra{\Psi} \text{ $\to$ } \rho_{AB}.$$
%and considers a mixed state $\rho_{AB}$ instead of the Bell state $\ket{\Psi}$, $$\ket{\Psi}\bra{\Psi} \text{ $\to$ } \rho_{AB},$$ usually called  
If $\ket{\Psi}\bra{\Psi}$ is affected by decoherence, the resulting resource state  $\rho_{AB}$ becomes
a Werner state with the
decoherence parameter $x$,
$$\rho_{AB}=(1-x)\ket{\Psi}\bra{\Psi}+x\frac{\mathbbm{1}_A\otimes\mathbbm{1}_B}{4}.$$
Therefore, this protocol is described by a depolarizing channel $\mathcal{D}_{a \to B}$ with distortion matrix equal to $\Lambda_{\mathcal{D}}$, Eq. \eqref{eq:depol}. Moreover, for an arbitrary resource $\rho_{AB}$, the standard  teleportation protocol can always be written as a Pauli channel $\mathcal{P}_{a \to B}$, Eq. \eqref{eq:Paulimap}. Another type of decoherence on $\ket{\Psi}\bra{\Psi}$ leads
to a teleportation channel described by the phase-flip channel, with distortion matrix given by \eqref{eq:phaseflip}. This 
 protocol will be called \textit{phase-flip noise} teleportation. Hence,  Eq. \eqref{eq:identitymap} describes the perfect teleportation protocol, while Eqs. \eqref{eq:phaseflip} and \eqref{eq:depol} are two different teleportation protocols that consider noise or decoherence affecting their resource state. 

 Fig. \ref{fig:QJSDvsTRaceforchoiandCD} shows that 
 for any decoherence parameter $x$
 the transmission distance $d_{\text{t}}^{\text{iso}}$
 between the perfect and the standard teleportation protocols with a Werner state as a resource, is greater than the
 distances to the phase-flip noise teleportation.
 %{\color{red} Editorial: this was not clear to me!  Please check if is is understandable now! \diego{Thanks for clarifying these paragraphs. I have included some minor suggestions in the text.} Perhaps we should add   
% 'phase-flip noise teleportation' into the caption of Fig 3. \diego{Done! Good idea :)}}
  
 %The transmission distance assigns similar values, but not equals, to \textbf{i)} the distance between the phase-flip noise SQT and SQT with a Werner resource, and to \textbf{ii)} the distinguishability between the PT protocol and the phase-flip noise SQT, while for the trace distance, the previous quantities are equal.
 %%
%%% above sentences look rather messy and not so important...

An analogous property holds also for the trace distance. 
In the case of the entropic channel divergence for $K=1$, the distance between the three different channels is equal, see Eq. \eqref{eq:entropicdistQT}, similar to the case of the trace distance, Eq. \eqref{eq:tracedistanceQT}.

The surfaces in Fig. \ref{fig:QJSDBallChoiIdent} and \ref{fig:QJSDBallECDIdent} can be interpreted now as
the standard teleportation protocols equally distant to the perfect one, represented by the vertex $c_1=c_2=c_3=1$.
The transmission distance
$d_\text{t}^{\text{iso}}$ 
between quantum channels is more restrictive regarding the values of the parameters $c_i$, than the entropic channel divergence and $d_\text{t}^{K=1}$,
which allows lower values for $c_i$.

\subsection{Distinguishing operations determined by Hamiltonians}\label{sec:applicationsHamil}

\noindent Several applications of quantum information theory involve the problem of distinguishing a particular Hamiltonian from a given set. For instance, to determine errors 
which occur by a real-life realisations of certain information
processing tasks. Other examples include identification
of a classical static force acting on a given quantum system \cite{Childs2000,Preskill2000}. Consider the distinguishability between two Hamiltonians $H_1$ and $H_2$, acting on a two-dimensional Hilbert space. 

Since three Pauli matrices, extended by the idenity matrix, 
 $\{\mathbbm{1},\vect{\sigma}\}$,
 form a Hilbert-Schmidt basis in the space of 
 Hermitian matrices of order two,
  any single-qubit Hamiltonian
 can represented by its Bloch vector,
\begin{align}\label{eq:defhm}
	H_m=h_m^0\mathbbm{1} + \vect{h}_m\cdot\vect{\sigma}.
\end{align}
The noiseless evolution of the state generated by a given Hamiltonian can be described by a unitary transformation, $\mathcal{U}_m(\rho)=U_m \rho U_m ^\intercal$, with
\begin{align}\label{eq:unitaryevol}
U_m=e^{-itH_m}=e^{-ith_m^0}\left(\cos t\mathbbm{1}  -i \sin t \vect{h}_m\cdot\vect{\sigma}\right)
\end{align}
where $\vect{h}_m\cdot\vect{h}_m=1$.

Making use of the Bloch form \eqref{eq:affinedecomp} of the unitary operation $\mathcal{U}_m$ we find the distortion matrix for both channels, 
\begin{align}\label{eq:affinedistinguishing}
	\Lambda_m&=\cos 2t (\mathbbm{1}-\vect{h}_m\vect{h}_m^\intercal)+\sin 2t [\vect{h}_m] + \vect{h}_m\vect{h}_m^\intercal=\nonumber\\
	&=e^{2t[\vect{h}_m]},
\end{align}
with $m=1,2$.
The symbol $[\vect{h}_m]$ denotes
the skew-symmetric matrix defined by $[\vect{h}_m] \vect{r} = \vect{h}_m\times \vect{r}$. This is evidently an unital operation and therefore its translation vector vanishes, $\vect{l}_m=\vect{0}$. 

To make the model more realistic assume that a single qubit, controlled by a Hamiltonian $H_m$, suffers decoherence induced by the \textit{depolarizing  channel.} The evolution of the system is governed by the master equation,
\begin{align}\label{eq:mastereq}
	\frac{\text{d}\rho}{\text{d}t}  = -i [H_m,\rho] - \Gamma (\rho -\frac{1}{2}\mathbbm{1}),
\end{align}
with  the damping rate $\Gamma$.
Adopting the convention $\hbar=1$
we assure that in these units
the frequency is equal to one.

Any Bloch vector $\vect{h}_m$ determines, through Eq. \eqref{eq:defhm}, the Hamiltonian $H_m$. Hence the master equation \eqref{eq:mastereq} leads to the following dynamics of the Bloch vector $\vect{r}$,
\begin{align}
	\frac{\text{d}\vect{r}}{\text{d}t}= 2 (\vect{h}_m\times \vect{r}) -\Gamma \vect{r}, 
\end{align}
where 
$\rho=\frac{1}{2}(\mathbbm{1}+\vect{r}\cdot\vect{\sigma})$.

Solving this equation, we arrive at the time dependence,
\begin{align}
	\vect{r}(t)=e^{-\Gamma t} e^{2t[\vect{h}_m]}\vect{r}_0.
\end{align}
The map $\mathcal{E}^{\text{dec}}_m$ can be written as a concatenation of a unitary dynamics and a depolarizing channel,  $\mathcal{E}_m^{\text{dec}}=\mathcal{D}\!\circ\!\mathcal{U}_m$, with the distortion matrix
\begin{align}\label{eq:affinedistinguishingdeco}
\Lambda_m^{\text{dec}}=e^{-\Gamma t} e^{2t[\vect{h}_m]}.
\end{align}
In Fig. \ref{fig:DecoherenceAndRotations}, we show
the resulting trajectories from these kinds of channels. We have fixed the initial Bloch vector, $\vect{r}_0=\frac{1}{\sqrt{3}}(1,1,1)^\intercal$ 
and evolved it by two Hamiltonians 
corresponding to $\vect{h}_1=(0,0,1)^\intercal$ and $\vect{h}_2=(1,0,0)^\intercal$. 
Note, how the combined channel (unitary transformation and depolarizing channel) becomes less distinguishable as the decoherence parameter $\Gamma$ increases.

Observe that a rotation of the vector $\vect{h}_m$ generates a particular transformation on the distortion matrix $\Lambda_m^{\text{dec}}$. Eq. \eqref{eq:affinedistinguishing} implies that $\tilde{\Lambda}_m^{\text{dec}}=R \Lambda_m^{\text{dec}} R^\intercal$ if $\vect{h}'_m=R \vect{h}_m$ with $R$ being an orthogonal matrix and $\tilde{\Lambda}_m^{\text{dec}}$ specified by $\vect{h}'_m$.

Assume that we need to distinguish between two Hamiltonians, $H_1$ and $H_2$, related to vectors $\vect{h}_1$ and $\vect{h}_2$, respectively. The evolved state of the system will depend on time and on the damping parameter $\Gamma$. A fundamental problem in quantum information is managing the decoherence effects while keeping measurement precision. Our aim is to find the optimal evolution time allowing one for the best distinguishability between both Hamiltonians in view of the transmission distance between the channels and the measures proposed in \cite{Raginsky2001,Childs2000}.
\begin{figure}
	\centering
	\includegraphics[width=.37\textheight]{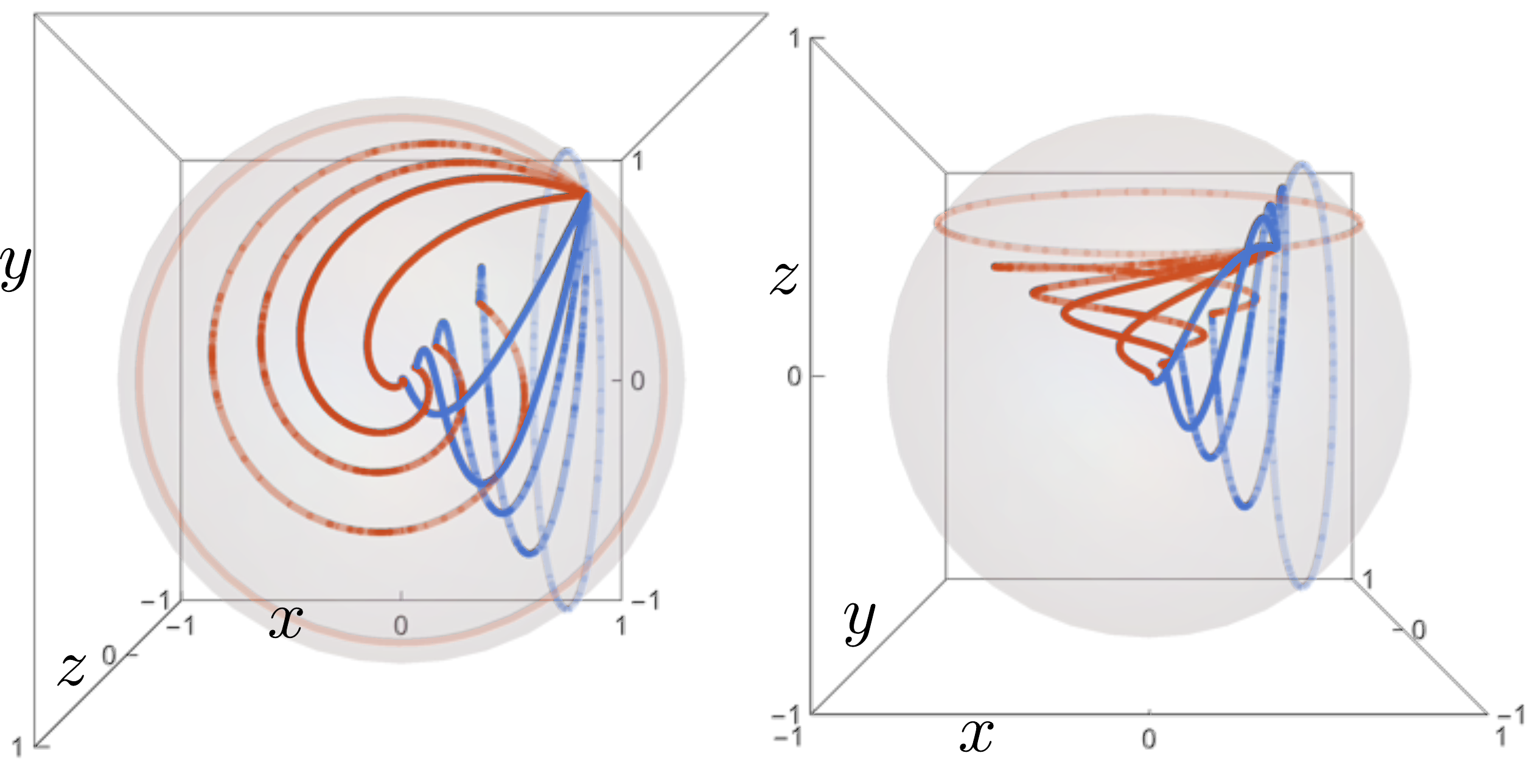}
	\caption{Visualizations of two unitary channels defined by $\vect{h}_1=(0,0,1)^\intercal$ and $\vect{h}_2=(1,0,0)^\intercal$, as a function of time
	$t\in \{0,\pi\}$ applied over an state with Bloch vector $\vect{r}_0=\frac{1}{\sqrt{3}}(1,1,1)^\intercal$, under a depolarizing channel with damping rate $\Gamma$. Each continues line is the trajectory of $\vect{r}_1(t)=e^{-\Gamma t} e^{2t[\vect{h}_1]}\vect{r}_0$ for different values of $\Gamma\in\{0,0.2,0.4,0.8,1\}$ (the opacity increase with $\Gamma$). The dashed lines correspond to the trajectories $\vect{r}_2(t)=e^{-\Gamma t} e^{2t[\vect{h}_2]}\vect{r}_0$ for the same values $\Gamma$. Both figures present the same trajectories from different perspectives.}
	\label{fig:DecoherenceAndRotations}
\end{figure}

\subsubsection{Comparison of distinguishability measures}

\noindent  We are going to analyze the transmission distance between quantum channels. For $N=2$, the Choi matrix of an arbitrary quantum channel $\mathcal{M}$ can be written as \cite{Shahbeigi2018},
\begin{align}\label{eq:ChoiState}
\rho_{\mathcal{M}}=\frac{1}{4}\left(\mathbbm{1}\otimes\mathbbm{1}+\mathbbm{1}\otimes \vect{l}\cdot\vect{\sigma}+\sum_{i,j}\Lambda'_{ij}\sigma_i\otimes\sigma_j\right),
\end{align}
where $\Lambda$ and $\vect{l}$ denote the distortion matrix and translation vector of the map, see Eq. \eqref{eq:affinedecomp}, while $\Lambda'_{ij}=(C\Lambda^\intercal)_{ij}$, with $C=\textrm{diag}(1,-1,1)$.

%{\color{red} Is the $T$ denoting  transposition the same as $T$ for the trace distance?} \diego{I changed the notation for the matrix $T=\textrm{diag}(1,-1,1)$ which was the same $T$ as the trace distance (transposition is denoted with $\intercal$). I took $C$ instead $T$ for the matrix $\textrm{diag}(1,-1,1)$, what do you think?}

Following Sec. \ref{sec:transmissiondistancebetQC}, we have to compare the evolved Choi states,
	\begin{align*}
		\rho_m=(\mathcal{E}^{\text{dec}}_m\otimes\mathbbm{1}) (\ket{\Phi}\bra{\Phi}),
	\end{align*}
where $\mathcal{E}_m^{\text{dec}}=\mathcal{D}\!\circ\!\mathcal{U}_m$. We use the transmission distance \eqref{eq:ChoiDisting}, which can be obtained by inserting Eq. \eqref{eq:affinedistinguishingdeco} into Eq. \eqref{eq:ChoiState} with $\vect{l}=0$. Note that  calculation of $d_\text{t}^{\text{iso}}(\mathcal{D}\!\circ\!\mathcal{U}_1,\mathcal{D}\!\circ\!\mathcal{U}_2)$ involves two non-commuting Choi states. 

Let us evaluate the entropic channel divergence \eqref{eq:entropicchanneldiv} for unital quantum channels \eqref{eq:affinedistinguishingdeco},
with distortion matrix proportional to a rotation matrix.  
\begin{proposition} 
	The entropic channel divergence \eqref{eq:entropicchanneldiv} between two unital maps $\mathcal{E}_1$ and $\mathcal{E}_2$,
	with distortion matrices $\Lambda_1=\alpha_1 R_1$
	and $\Lambda_2=\alpha_2 R_2$, respectively 
	%	$i=1,2$, being $R_i$ rotation matrices and $\alpha_i \in \mathbb{R}_{\geq 0}$, results
	  reads
	\begin{align}\label{eq:ECD1}
	d_{\text{t}}^{K=1}\left(\mathcal{E}_1,\mathcal{E}_2\right)=\sqrt{f\!\left( r_{\text{opt}} \right)-\frac{[f\!\left(\alpha_1^2 \right)+f\!\left(\alpha_2^2 \right)]}{2}},
	\end{align}
	with
	\begin{align}\label{eq:ropt}
	r_{\text{opt}}=\alpha_1^2+\alpha_2^2+
	\alpha_1\alpha_2
	\left({\rm Tr}{[\Lambda_{1}^\intercal\Lambda_{2}]}-1\right),
	\end{align}
 and the function $f(\cdot)$ defined in Eq. \eqref{eq:functionf}.
\end{proposition}

\begin{proof}
	Employing the same reasoning used to derive 
	Eq.~\eqref{eq:extremumQJSD}, we arrive at,
	\begin{align}
	\text{QJSD}\!\left[\mathcal{E}_1(\rho),\mathcal{E}_2(\rho)\right]&=f\left[\vect{r}\cdot(\Lambda_{\overline{\mathcal{E}}}^\intercal \Lambda_{\overline{\mathcal{E}}}) \vect{r} \right]-\frac{1}{2}f\!\left(\alpha_1^2 \right)+\nonumber\\
	&-\frac{1}{2}f\!\left(\alpha_2^2 \right), \label{eq:QJSDdisting}
	\end{align}
	where
	%$f\!(\cdot)$ is defined in Eq. \eqref{eq:functionf} and 
	$\Lambda_{\overline{\mathcal{E}}}=\alpha_1 R_1+\alpha_2 R_2$. To calculate the entropic channel divergence we need to optimize the function $f$ used in Eq. \eqref{eq:functionf}, 
	$$f\left[\vect{r}\cdot(\Lambda_{\overline{\mathcal{E}}}^\intercal \Lambda_{\overline{\mathcal{E}}}) \vect{r} \right].$$
	As $f\!(x)$ is a decreasing function of $x$ in $[0,1]$, we have to minimize 
	\begin{align}\label{eq:addEq}
	\vect{r}\cdot(\Lambda_{\overline{\mathcal{E}}}^\intercal \Lambda_{\overline{\mathcal{E}}}) \vect{r}=\alpha_1^2+\alpha_2^2+2p_1p_2 \vect{r}\cdot(\Lambda_{1}^\intercal \Lambda_{2}) \vect{r},
	\end{align}
	over the sphere $\vect{r}\cdot\vect{r}=1$.
	
	Taking $\Lambda_{1}^\intercal\Lambda_{2}=e^{\phi[\vect{h}]}$ for some $\vect{h}$ such that $\vect{h}^2=1$, see Eq. \eqref{eq:affinedistinguishing}, we find that
	\begin{align*}
	\vect{r}\cdot(\Lambda_{1}^\intercal \Lambda_{2}) \vect{r}=\cos\phi + (1-\cos\phi)(\cos{\gamma})^2,
	\end{align*}
	where $\cos{\gamma}=\vect{h}\cdot \vect{r}$.
	The minimum of the function
	$\vect{r}\cdot(\Lambda_{1}^\intercal \Lambda_{2}) \vect{r}$ in the sphere $\vect{r}^2=1$ is correspondingly given by the minimum of the previous function over the parameter $\gamma$. It is straightforward to show that $\gamma=\pi/2$ minimizes $\vect{r}\cdot(\Lambda_{1}^\intercal \Lambda_{2}) \vect{r}$, and therefore,
	\begin{align*}
	\min_{\vect{r}}\left\{\vect{r}\cdot e^{\phi[\vect{h}]} \vect{r}\right\}=\cos \phi.
	\end{align*}
	Finally, employing the following equality,
	\begin{align*}
	\cos\phi=\frac{\Trr{e^{\phi[\vect{h}]}}-1}{2},
	\end{align*}
	we arrive at, 
	\begin{align}\label{eq:minr}
	\min_{\vect{r}}\left\{\vect{r}\cdot(\Lambda_{1}^\intercal \Lambda_{2}) \vect{r}\right\}=\frac{\Trr{\Lambda_{1}^\intercal\Lambda_{2}}-1}{2}.
	\end{align}
	By inserting this in Eq. \eqref{eq:addEq}, we obtain Eq. \eqref{eq:ECD1}.
\end{proof}

On the other hand, if $\vect{h}'_m=R \vect{h}_m$,
with $R$ denoting an orthogonal matrix of order three, the corresponding affine matrix $\Lambda_m^{\text{dec}}$
transforms as 
\begin{align}\label{eq:unitarytransf}
	\tilde{\Lambda}_m^{\text{dec}}=R \Lambda_m^{\text{dec}} R^\intercal .
\end{align}
Therefore, the quantum operation $\tilde{\mathcal{E}}^{\text{dec}}_m$ associated with $\tilde{\Lambda}_m^{\text{dec}}$ can be written as $\tilde{\mathcal{E}}^{\text{dec}}_m=\mathcal{R}\!\circ\!\mathcal{E}^{\text{dec}}_m\!\circ\!\mathcal{R}^{-1}$, where $\mathcal{R}$ is the unitary channel corresponding to the rotation
matrix $R$, while $\mathcal{E^{\text{dec}}}_m$ is determined by $\Lambda_m^{\text{dec}}$. 

Since the distance measures between quantum operations satisfy the unitary invariance \eqref{eq:unitaryinv},
the distinguishability between operations $\mathcal{E}^{\text{dec}}_1$ and $\mathcal{E}^{\text{dec}}_2$ specified by $\vect{h}_1$ and $\vect{h}_2$, respectively, depends only on the angle 
$$\theta=\arccos \vect{h}_1 \cdot \vect{h}_2,$$ 
and the damping rate $\Gamma$. Thus, without losing generality we can fix the vector $\vect{h}_1$ in 
the direction $z$. 

In Fig. \ref{fig:dJSvsBurPURE}, we present the transmission distance $d_\text{t}^{\text{iso}}(\mathcal{U}_1,\mathcal{U}_2)$ and the Bures distance $D_B(\mathcal{U}_1,\mathcal{U}_2)$ defined in \eqref{eq:buresdistanceChoi}.
Both quantities are computed in the noiseless case, $\Gamma=0$,
and shown as functions of time $t$
for different values of the angle $\theta$. 
In this case the entropic channel divergence
is equal to the transmission distance \eqref{eq:ChoiDisting} between quantum channels. The Bures distance $D_B$ is based on the quantum fidelity between the Choi matrices -- see \cite{Raginsky2001,Childs2000}.

The time in which the distinguishability is maximal, according to  the measures analyzed, reads
\begin{align}
	t_{\text{max}}= \left\{ \begin{matrix}
		\pi/2 &&\text{if } \cos \theta \geq 0 \\
	\frac{1}{2} \cos ^{-1}\left(\frac{\cos \theta +1}{\cos \theta -1}\right) &&\text{if } \cos \theta < 0
	\end{matrix}\right. . \label{eq:noiselesstimes}
\end{align} 
Note that if  $\cos \theta > 0$,
both unitary operations cannot be distinguished with probability one at any time.  However, if $\cos \theta \leq 0$, there exists a time in which the pure Choi states are orthogonal and can be perfectly distinguished at the selected interaction time $t_\text{max}$.

\begin{figure}
	\centering
	\includegraphics[width=.37\textheight]{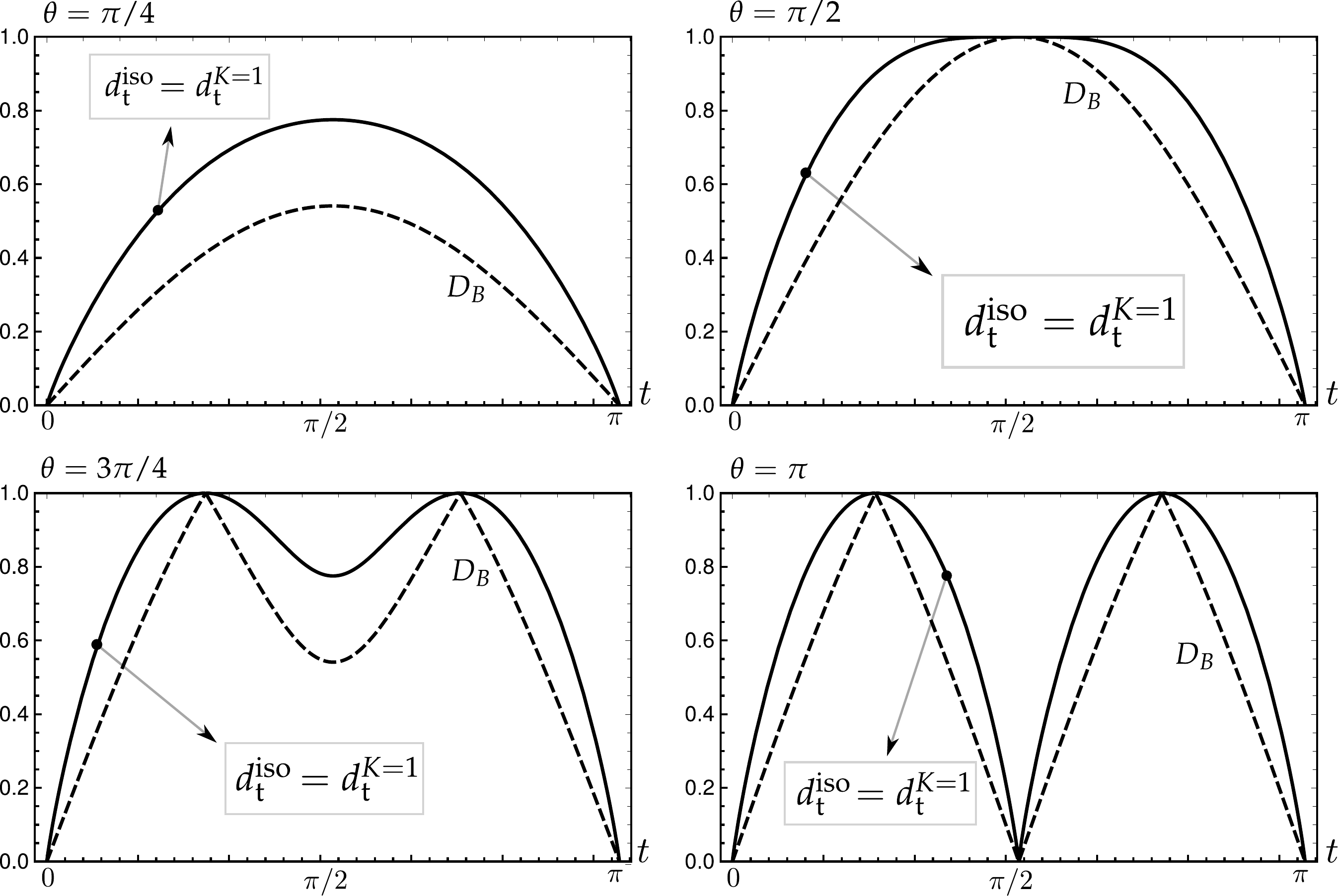}
	\caption{Transmission distance \eqref{eq:ChoiDisting}, Bures distance \eqref{eq:buresdistanceChoi}, and the entropic channel divergence \eqref{eq:entropicchanneldiv}, between two unitary operations \eqref{eq:unitaryevol}, whose corresponding vectors $\vect{h}_1$ and $\vect{h}_2$ form an angle $\theta\in\{\pi/4,\pi/2,3\pi/4,\pi\}$, as functions of time $t$. We assume that $\hbar=\omega=1$,
	so all quantities are dimensionless.}
	\label{fig:dJSvsBurPURE}
\end{figure}

Let us take into account effects of the decoherence. The depolarizing channel \eqref{eq:affinedistinguishingdeco}, transforms the original unitary rotations into channels that send states closer to the maximally mixed state -- see Fig. \ref{fig:DecoherenceAndRotations} -- so the problem of distinguishability between the channels becomes more difficult.

This problem was already treated in Ref. \cite{Childs2000}, where it was suggested to select a constant initial state, with the Bloch vector $\vect{r}_0=(1,0,0)^\intercal$, and to choose the optimal time as the one minimizing the error probability $P_{\text{error}}$.
Such an optimal time $t_{\text{opt}}$ corresponds
to the maximal distinguishability between both
evolved states, 
\begin{align}\label{eq:errprob}
P_{\text{error}}=\frac{1}{2}[1-\exp(-pt)\left|\sin t \right|].
\end{align}
At a time $t_{\text{opt}}=\arctan(1/p)$, $P_{\text{err}}$ is minimized and thus the information gained by the measurement is maximized.

Regarding  entropic distinguishability measures, Fig. \ref{fig:DecoherenceTs} displays behaviour of the transmission distance under unitary evolution and decoherence, for angle $\theta=\pi/2$ and exemplary values of the damping rate, $\Gamma\in \{0, 0.3 , 0.6 , 0.9 , 1.2 , 1.5 , 1.8\}$.

\begin{figure}
	\centering
	\includegraphics[width=.35\textheight]{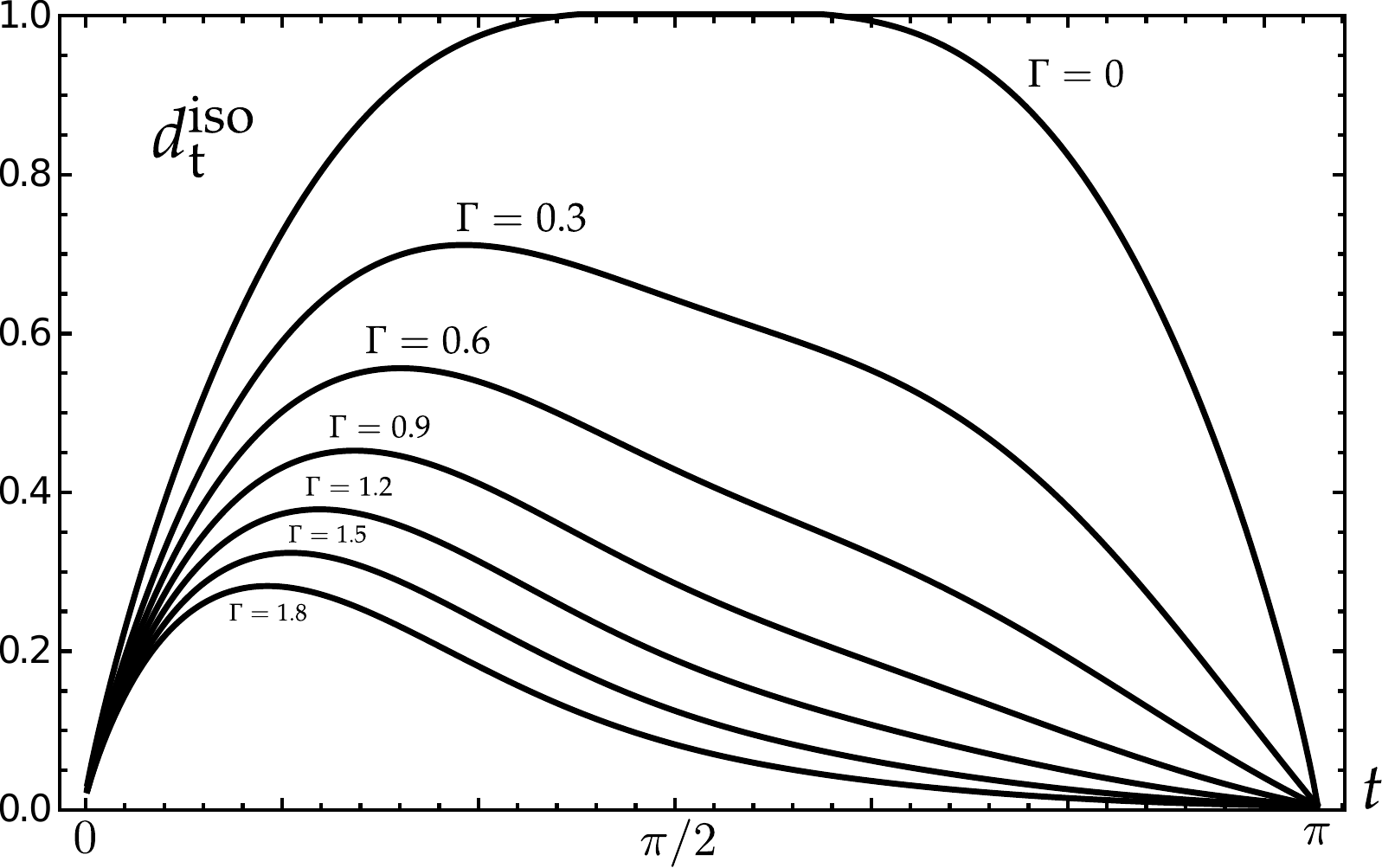}
	\caption{Transmission distance $d_{\text{t}}^{\text{iso}}(\mathcal{D} \circ \mathcal{U}_1,\mathcal{D} \circ \mathcal{U}_2)$ as a function of time $t$, where the affine decomposition of the maps $\mathcal{D}\circ\mathcal{U}_i$ is given by \eqref{eq:affinedistinguishingdeco}. The angle
	between the Bloch vectors defining both Hamiltonians ($\vect{h}_1$ and $\vect{h}_2$) is $\theta=\pi/2$. Here $\mathcal{D}$ denotes the depolarizing channel with damping rate $\Gamma$, which labels the curves. The larger damping rate, the shorter time $t_\text{max}$ of maximal distinguishability.}
	\label{fig:DecoherenceTs}
\end{figure}

The entropic channel divergence is given by taking $\alpha_i=e^{-\Gamma t}$ and $\Lambda_i=e^{2t[\vect{h}_i]}$, with $i=1,2$, in Eq. \eqref{eq:ECD1}. 
In this way one obtains,
	\begin{align}
		\Trr{\Lambda_{1}^\intercal\Lambda_{2}}\!&=\!2 \cos (2 \theta ) \sin ^4\!(t)\!+\!2 \cos(\theta ) \sin ^2\!(2 t)\!+\!\cos(2 t)+\nonumber\\&+\frac{3}{4} \cos (4 t)+\frac{5}{4},
	\end{align}
\noindent where $\theta$ denotes the angle between 
both Bloch vectors,
$\vect{h}_1$ and $\vect{h}_2$. Inserting \eqref{eq:minr} into \eqref{eq:QJSDdisting}, we arrive at the dependence of the entropic channel divergence on the angle $\theta$, the time $t$ and the damping parameter $\Gamma$.

One can pose a natural question, which interaction time is optimal to distinguish Hamiltonians $H_1$ and $H_2$ under decoherence? In the noiseless situation $\Gamma=0$, the entropic channel divergence results to be equal to the transmission distance between quantum channels, Eq. \eqref{eq:ChoiDisting}, therefore, the interaction time \eqref{eq:noiselesstimes} is  optimal for this measure as well. In presence of decoherence, each distinguishability measure has its own behavior, leading to different values of optimal interaction times. Fig. \ref{fig:times} shows that the best times to measure the distinguishability related to the transmission distance  $d_{\text{t}}^{\text{iso}}$ 
%and $d_{\text{t}}^{K=1}$) 
are shorter than those arising from minimizing the error probability of distinguishing the two evolved states \eqref{eq:errprob}, proposed in \cite{Childs2000}.

\begin{figure}
	\centering
	\includegraphics[width=.35\textheight]{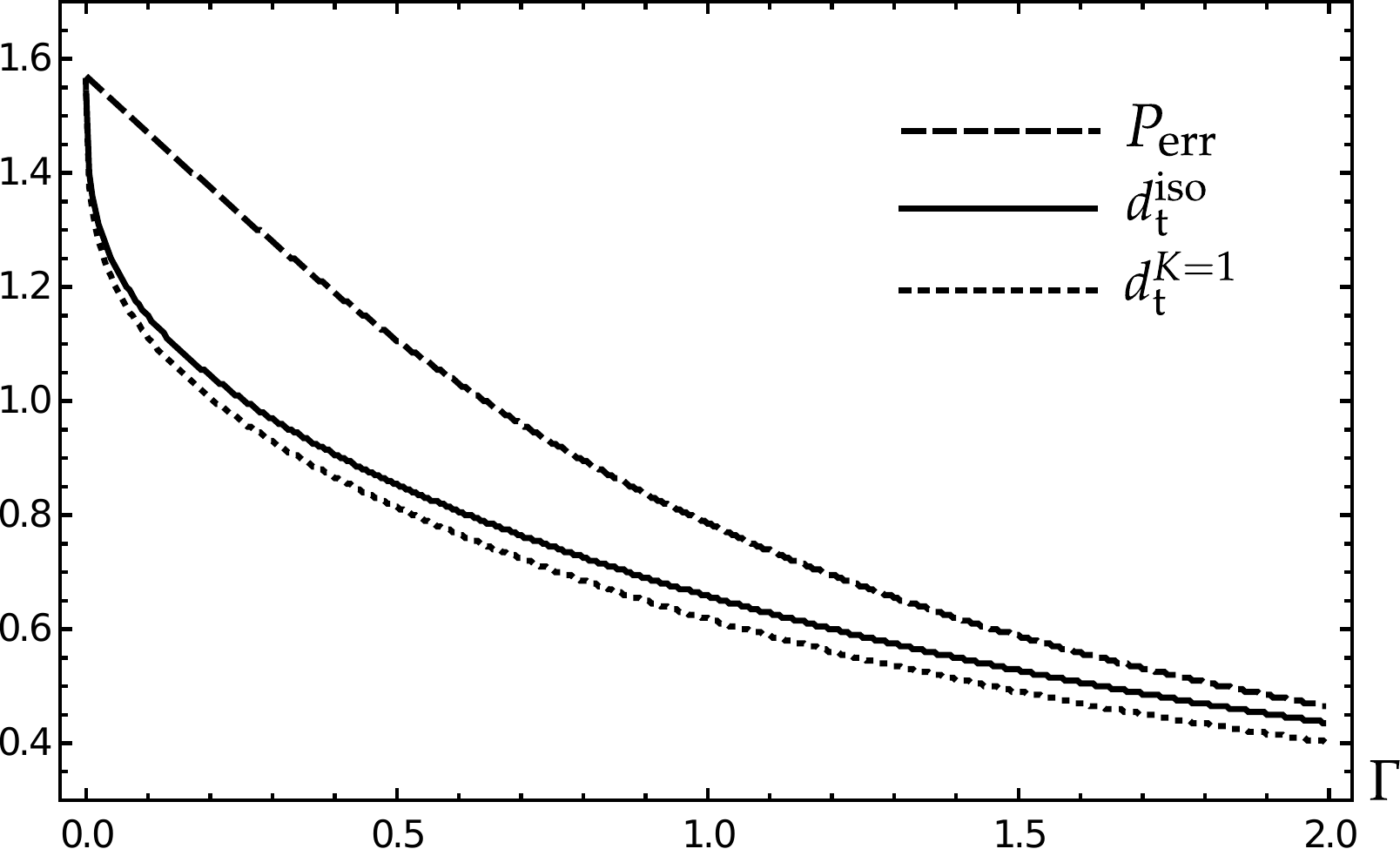}
	\caption{Optimal times as a function of the noise parameter $\Gamma$, in the distinguishability of Hamiltonians, see Sec. \ref{sec:applicationsHamil}. The maps are given by  \eqref{eq:affinedistinguishingdeco}. The dashed line corresponds to optimal times for the probability of error, $P_{\text{err}}$, see \eqref{eq:errprob}. Continuous line
	represents the transmission distance $d_{\text{t}}^{\text{iso}}$
	between quantum channels \eqref{eq:ChoiDisting}, while the dotted line corresponds to the optimal times in the case of the entropic channel divergence, $d_{\text{t}}^{K=1}$,
	see  \eqref{eq:entropicchanneldiv}.}
	\label{fig:times}
\end{figure}

\section{Concluding remarks}\label{sec:concludingR}

\noindent We \diego{have} introduced two entropic measures of distinguishability between quantum operations
using the square root of the quantum Jensen-Shannon divergence, also called \textit{transmission distance}. We have investigated their properties and physical interpretations. 

In the case of the \textit{transmission distance between quantum channels} $d_{\text{t}}^{\text{iso}}$, 
we have shown that this measure satisfies several criteria for a suitable distance measure between maps.
Even though this quantity does not satisfy the chaining property, this is the case if one of the maps applied first
is bistochastic, which is a key property for estimating errors in quantum information protocols \cite{Gilchrist2005}. Furthermore, the transmission distance between quantum channels does not require any optimization procedure and it can be directly obtained by calculating the entropy of a map, defined in \cite{Roga2011}. Regarding the physical interpretation of this measure, $d^{\text{iso}}_{\text{t}}$ is the dense coding capacity for a noiseless dense coding protocol. 
It is therefore fair to expect
that the transmission distance between quantum channels is a good candidate for error or diagnostic measures.

In Sec. \ref{sec:entropicdisting}, we \diego{have} introduced the \textit{entropic channel divergence} $d_{\text{t}}^{K}$,
parameterized by the size $K$ of the ancilla.
%exploring its main mathematical properties and its suitability as a distance quantifier between quantum operations. 
In addition to the requirements mentioned in \cite{Raginsky2001,Gilchrist2005}, we have shown that $d_{\text{t}}^{K}$ satisfies the \textit{chain rule}. 
This property allows one to prove the amortization collapse of the entropic channel divergence,
which can be useful to obtain new single-letter converse bounds on the capacity of adaptive protocols in channel discrimination theory \cite{Wilde2020a}. Regarding physical motivation, $d_{\text{t}}^{K}$ is the square root of the quantum reading capacity in the equiprobable case \cite{Pirandola2011}, and it can be identified as the capacity of a dense coding protocol with a resource influenced by decoherence \cite{Laurenza2020}.

In Sec. \ref{sec:connection}, we have considered the case of Choi-stretchable channels. For these kinds of quantum operations, $d_{\text{t}}^{\text{iso}}$ and $d_{\text{t}}^{N}$ are equal, establishing a particular situation, in which the transmission distance between quantum operations is equal to the stabilized entropic channel divergence \eqref{eq:entropicchanneldivstab}.

To demonstrate the analyzed measures in action,
we \diego{have} investigated the distinguishability of two Pauli channels
and provided analytical expressions for the distance $d_{\text{t}}^{\text{iso}}$
%(equal to $d_{\text{t}}^{K=2}$)
and the entropic divergence $d_{\text{t}}^{K=1}$. 
As the standard teleportation protocol can be written as a Pauli map, we \diego{have studied} the presence of noise in quantum teleportation by calculating both distinguishability measures. 
The transmission distance $d_{\text{t}}^{\text{iso}}$ between quantum channels occurred to be the most sensitive to decoherence, while the trace distance between the corresponding Choi states is more sensitive than the entropic channel divergence.
%and the trace distance channel divergence, defined in Eq. \eqref{eq:channeldivTr}.

In the case of a Hamiltonian evolution under decoherence, we \diego{have}
compared the distance $d_{\text{t}}^{\text{iso}}$ and the divergence $d_{\text{t}}^{K=1}$ between the quantum operations with the Bures distance between the corresponding Choi states
and the \textit{probability of error}, originally studied \cite{Childs2000}. In the absence of noise, the distance measures defined by employing the transmission distance 
become equal, $d_\text{t}^{\text{iso}}=d_\text{t}^{K=1}$, showing a smoother behaviour than the Bures distance and exhibiting equal times of maximal distinguishability. 

To distinguish between dynamics generated
by two Hamiltonians subjected to decoherence,
we have studied the entropic measures $d_{\text{t}}^{\text{iso}}$ and $d_{\text{t}}^{K=1}$
 and compared them with the error  probability $P_{\text{err}}$. %proposed in Ref. \cite{Childs2000}, 
For these measures we identified the time window
of maximal distinguishability 
while varying the decoherence rate $\Gamma$.
The above observations suggest that the measures of the distance between quantum operations based on the square root of the Jensen-Shannon divergence
(in this case equivalent  to the Holevo quantity) introduced in this work will find their applications in further theoretical and experimental studies.

\section*{Acknowledgments}
\noindent D.G.B. and P.W.L. 
are grateful to the Jagiellonian University for the hospitality during their stay in Cracow.
They acknowledge financial support by Consejo Nacional de Investigaciones Científicas y Técnicas (CONICET), 
and by 
%Argentina. D.G.B. has a fellowship from CONICET. D.G.B. and P.W.L. are also grateful to
Universidad Nacional de Córdoba (UNC), Argentina.
K.{\.Z}. is supported by Narodowe Centrum Nauki under the Quantera project number 2021/03/Y/ST2/00193 
and by Foundation for Polish Science 
under the Team-Net project no. POIR.04.04.00-00-17C1/18-00.

\section{Appendix}
	
\subsection{Channel divergence with trace distance between unital channels} \label{sec:tracedistanceappendix}
\noindent	Let us calculate
	$$d_{\text{Tr}}^{K=1}(\mathcal{E},\mathcal{F})=\sup_{\rho\in\mathcal{M}_N}T[\mathcal{N}(\rho),\mathcal{M}(\rho)],$$
	for two arbitrary unital quantum operations $\mathcal{N}$ and $\mathcal{M}$ with $N=2$, being T$(\cdot,\cdot)$ the trace distance.
	
	Performing  required calculations we arrive at an expression,
		\begin{align}
			d_{\text{Tr}}^{K=1}(\mathcal{N},\mathcal{M})= \frac{1}{2}\max_{\vect{r}}\sqrt{\vect{r}\cdot \Delta \vect{r}},
		\end{align}
	 where $\vect{r}$ denotes the Bloch vector $\rho$ and $\Delta=(\Lambda_\mathcal{N}-\Lambda_\mathcal{M})^{\intercal}(\Lambda_\mathcal{N}-\Lambda_\mathcal{M})$.
	
	We need now to optimize $\sqrt{\vect{r}\cdot \Delta \vect{r}}$ over the sphere $\vect{r}^2=1$. As $\Delta$ is a symmetric positive square matrix, we can take its spectral decomposition,
		\begin{align}
			\Delta=\sum_i \lambda_i^\Delta \vect{x}_i \vect{x}_i^\intercal,
		\end{align}
		where $\vect{x}_i$ denotes the eigenvector of $\Delta$ corresponding to the eigenvalue $\lambda_i$.
		One obtains, therefore,
		\begin{align*}
			\vect{r}\cdot\Delta\vect{r}=\sum_i \lambda_i^\Delta (\vect{r}\cdot \vect{x}_i)^2. 
		\end{align*}
		Having in mind that $\lambda_i^\Delta\geq 0$ and $(\vect{r}\cdot \vect{x}_i)^2 \in [0,1]$ for any $i$, it is clear that the maximum is achieved when $\vect{r}=\vect{x}_k$ with $k$ such that $\lambda_k^{\Delta}\geq \lambda_i^{\Delta}$ for all $i$. This implies directly Eq. \eqref{eq:channeldivTrUnital},  specifically,		
		\begin{align*}
		d_{\text{Tr}}^{K=1}(\mathcal{N},\mathcal{M})=\frac{1}{2}\max_i \sqrt{\lambda_i^{\Delta}},
		\end{align*}
		where $\{\lambda_i^{\Delta}\}_i$ is the set of eigenvalues of the matrix $$\Delta=(\Lambda_\mathcal{N}-\Lambda_\mathcal{M})^{\intercal}(\Lambda_\mathcal{N}-\Lambda_\mathcal{M}).	$$ 
	
%\subsection{Trance distance distinguishing Hamiltonians}
%	Analysis of plots (Choi version and channel div. version have the same optimal times!)

\subsection{Upper bound for the transmission distance}\label{sec:appendixBounds}
\noindent %In the literature can be found 

\noindent Two different upper bounds for the transmission distance $d_{\text{t}}$ between quantum states can be found in the literature. One in terms of the entropic distance $D_{E}$ defined in Eq. \eqref{eq:boundK} \cite{Lamberti2008}
and the other one based on  the  square root of the trace distance $\sqrt{\text{T}}$ \cite{Briet2009}.

In Eq. \eqref{eq:bounds} we have included the corresponding bound for quantum maps,
\diego{\begin{align}
	\label{last_one}
	d_{\text{t}}^{\text{iso}}(\mathcal{E},\mathcal{F}) \leq \min \left\{\sqrt{ T(\mathcal{E},\mathcal{F})} \ , \ D_E(\mathcal{E},\mathcal{F})\right\}.
	\end{align}}
Note that the function {\sl minimum} appears in this bound. 
In Fig. \ref{fig:bounds}, we analyze an ensemble of random pairs of Choi states of order four, 
corresponding to unital Pauli maps,
and compared the distances given by $\sqrt{T}$ and $D_E$ between them. Numerical results show that 
for some pairs of channels it holds $\sqrt{T}>D_E$ and for others $\sqrt{T}<D_E$. 
These observations imply that using the
function {\sl minimum} in Eq. \eqref{last_one}
is justified as it  makes the upper bound stronger.

\begin{figure}
	\centering
	\includegraphics[width=.35\textheight]{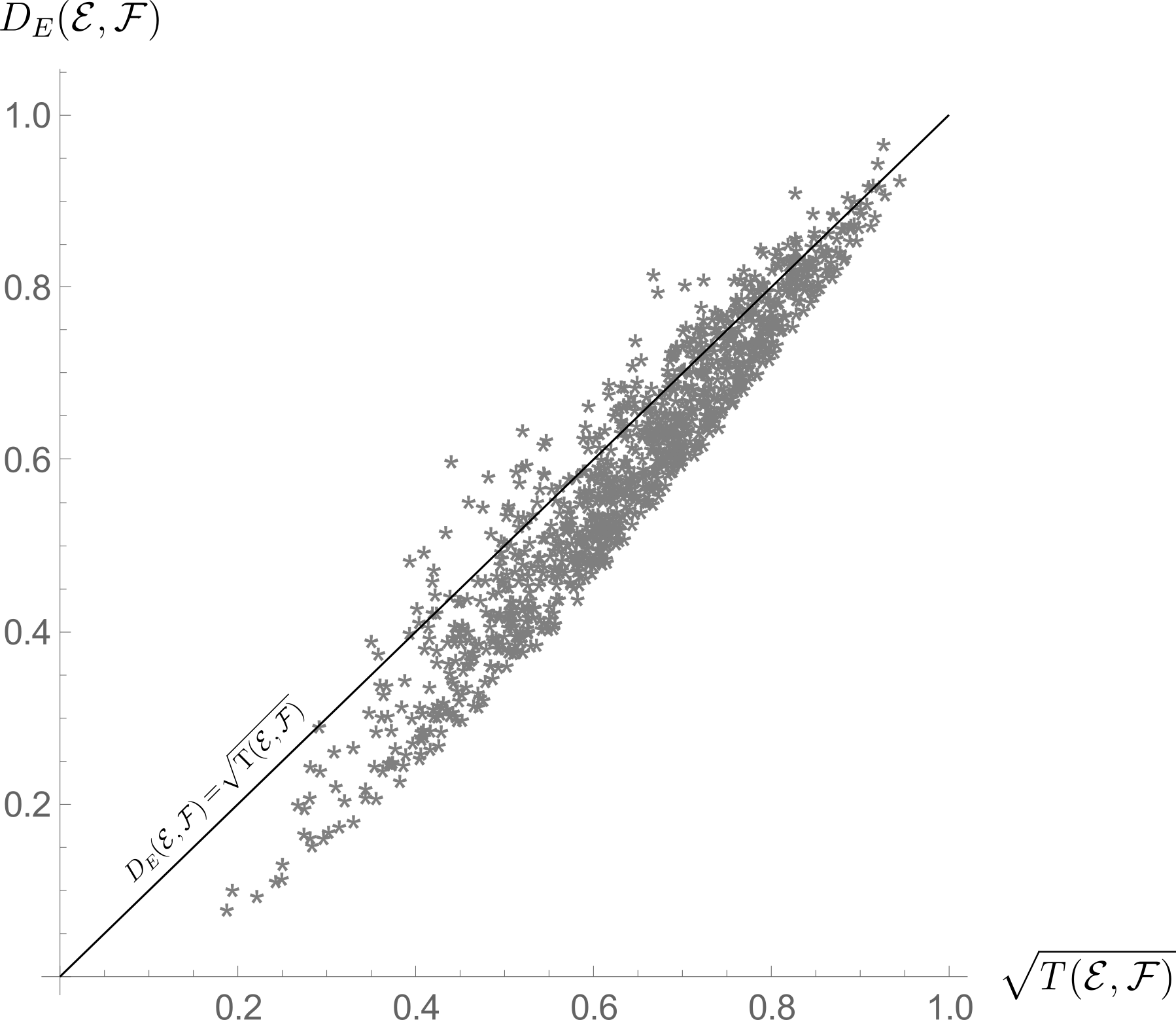}
	\caption{
	Square root of the trace distance
	$\sqrt{T(\mathcal{E},\mathcal{F})}$ 
	between random  Choi states,
	and their entropic distance
	 $D_E(\mathcal{E},\mathcal{F})$, defined in Eqs. \eqref{eq:tracedistanceChoi} and \eqref{eq:KarolmeasureChoi}, between 1000 
	   pairs of channels
	   taken randomly according to the flat measure
	  in the regular tetrahedron of Pauli channels.
	   As points are scattered on both sides of the diagonal, these results show that the 
	     {\sl min} function should be used in the upper bound
	   \eqref{last_one}.
	   }
	\label{fig:bounds}
\end{figure}

\bibliographystyle{ieeetrv7}
\bibliography{libraryDEF}

\end{document}